\documentclass[a4paper,UKenglish,cleveref, autoref, thm-restate]{lipics-v2021}
\usepackage{graphicx} 
\usepackage{algpseudocodex}
\usepackage{algorithm}
\usepackage{cleveref}
\usepackage{hyperref}
\hypersetup{colorlinks=true,linkcolor=blue,urlcolor=blue,citecolor=blue}
\usepackage{framed}
\usepackage{amsmath,amsfonts,amssymb}
\newtheorem{question}{Question}
\usepackage{lmodern} 
\DeclareFontShape{T1}{lmr}{m}{scit} { <-> ssub * cmr/m/sc }{}

\usepackage{placeins} 

\usepackage[appendix=inline,bibliography=common]{apxproof} 
\newtheorem{thm}{Theorem}

\newtheorem{defn}[thm]{Definition}
\newtheorem{obs}[thm]{Observation}
   \newtheoremrep{theorem}[thm]{Theorem}
 \newtheoremrep{lemma}[thm]{Lemma}
 \newtheoremrep{corollary}[thm]{Corollary}
 \newtheoremrep{observation}[thm]{Observation}

\newtoggle{sketch}
\togglefalse{sketch}
\makeatletter
 \@ifpackagewith{apxproof}{appendix=inline}{}{}
 \ifthenelse{\equal{\axp@appendix}{strip}}{}{}
 \NewEnviron{sketch}
 { \ifthenelse{\equal{\axp@appendix}{inline}}{}{
     \iftoggle{sketch}{
     \begin{proofsketch}
       \BODY
     \end{proofsketch}
       }{}
   }}
 \makeatother

\newcommand{\apxblue}[1]{\textcolor{black}{#1}}

\usepackage[shortlabels,inline]{enumitem} 
\usepackage{ulem}
\usepackage{silence}
\newcommand{\bconcat}[0]{\smile}

\nolinenumbers

\title{\texorpdfstring{$O(n)+f(k)$}{O(n)+f(k)}: Truly Linear FPT}

\author{Benjamin Merlin Bumpus}{Instituto de Matemática e Estatística, Universidade de São Paulo. Rua do Matão, 1010 — 05508–090, São
Paulo, SP, Brasil.}{bumpus@usp.br}{https://orcid.org/0000-0002-8686-2319}{}
\author{Rod Downey}{Victoria University of Wellington, School of Mathematics and Statistics, PO Box 600 Wellington, New Zealand}{rod.downey@vuw.ac.nz}{https://orcid.org/0000-0003-4381-2845}{}
\author{Tala Eagling-Vose}{Department of Computer Science, Durham University, Durham, UK}{tala.j.eagling-vose@durham.ac.uk}{https://orcid.org/0009-0008-0346-7032}{}
\author{Jessica Enright}{School of Computing Science, University of Glasgow}{jessica.enright@glasgow.ac.uk}{0000-0002-0266-3292}{Partly supported by EPSRC grants  EP/T004878/1 and EP/V032305/1}
\author{Michael R. Fellows}{Institute of Informatics, University of Bergen (Emeritus), Bergen, Norway}{michael.fellows@uib.no}{https://orcid.org/0000-0002-6148-9212}{}
\author{David C. Kutner}{Department of Computer Science, Durham University, Durham, UK}{david.c.kutner@durham.ac.uk}{https://orcid.org/0000-0003-2979-4513}{Partly supported by EPSRC grant EP/T004878/1, \emph{Multilayer Algorithmics to Leverage Graph Structure (MultilayerALGS)}.}
\author{Laura Larios-Jones}{School of Computing Science, University of Glasgow, Glasgow, UK}{l.larios-jones.1@research.gla.ac.uk}{0000-0003-3322-0176}{}
\author{Barnaby Martin}{Department of Computer Science, Durham University, Durham, UK}{barnaby.d.martin@durham.ac.uk}{https://orcid.org/0000-0002-4642-8614}{supported by Leverhulme Trust Research Project Grant RPG-2024-182.}
\author{Frances Rosamond}{Institute of Informatics, University of Bergen (Emerita), Bergen, Norway}{frances.rosamond4@gmail.com}{https://orcid.org/0000-0002-5097-9929}{}
\author{Ella Yates}{School of Computing Science, University of Glasgow, Glasgow, UK}{3032195Y@student.gla.ac.uk}{https://orcid.org/0009-0000-7979-0401}{}

\authorrunning{Bumpus, Downey, Eagling-Vose, Enright, Fellows, Kutner, Larios-Jones, Martin, Rosamond, and Yates} 

\Copyright{Benjamin Bumpus, Rod Downey, Tala Eagling-Vose, Jessica Enright, Michael R. Fellows, David C. Kutner, Laura Larios-Jones, Barnaby Martin, Frances Rosamond, and Ella Yates}



\date{}

\begin{document}

\maketitle

\begin{abstract}
    Parameterized complexity has always been concerned with practical computing: by confining combinatorial explosion to a secondary parameter $k$, one can uncover why and how many NP-hard problems are effectively tackled in practice.  Today, however, the scale of data has changed: scientists study Big Data, which is so large that even quadratic dependence in the total input size $n$ is unaffordable. Therefore, what constitutes a practical algorithm has also changed.   Classically, parameterized complexity is blind to the difference between defining fixed parameter tractability multiplicatively (i.e.~$f(k) \cdot n^c$) or additively (i.e.~$f(k) + n^c$).
    But what if the constant $c$ is one and we require true linearity, is this distinction still inconsequential? 
    Here, we define and explore Truly Linear FPT (TLFPT) -- that is $O(n)+f(k)$ -- and show that it is a strict subset of Linear FPT (LFPT) -- that is $O(n) \cdot f(k)$ -- via diagonalization. 
    
     Populating TLFPT requires careful consideration of linear-time algorithmics and data structures. We meet many inhabitants of TLFPT: {\sc SAT, Vertex Cover, Min-Max Matching, $(n-k)$-Coloring, Diverse Pair of Matchings, $k$-Path}, and {\sc $H$-Coloring}. 
     Our parameterizations are equally varied. Beyond classical parameters like solution size, we leverage two parameters, treedepth and BFS-width, which are particularly well-suited to the TLFPT regime.  
     We do so by developing techniques based on depth- and breadth-first search.

    For parameterized complexity to be of service to the scientific community, we need to contend with Big Data. For sufficiently large inputs, FPT beyond linear may not suffice. 
    Thus, there is a practical and theoretical need for more ambitious goals. 
    TLFPT is a first step forward.
    \end{abstract}
    
    \bigskip

    \vfill

    {\small \noindent For the purpose of open access, the author(s) has applied a Creative Commons Attribution (CC BY) license to any Author Accepted Manuscript version arising from this submission. No data were created or analyzed in this work.}

\newpage

\section{Introduction}\label{sec:intro}

There is a simple premise at the core of parameterized complexity: it is unrealistic to quantify the running time of algorithms solely in terms of the total input size~$n$. In many cases there is a smaller, problem-relevant quantity $k$ to which combinatorial explosion can be confined. From this perspective, an efficient algorithm is one running in time $f(k) \cdot n^{c}$, where $f$ is a function independent of $n$ and the constant $c$ is independent of both $n$ and $k$. This is \textit{fixed-parameter tractability} (FPT). 

This guiding philosophy is far from a mere mathematical curiosity; it bears genuine practical weight. Even for modest values of $n$ and $k$, the difference between an $n^{k}$ running time and a $2^{k} n$ running time can span dozens of orders of magnitude~\cite{downey2013fundamentals}. In other words, isolating the exponential dependence within~$k$ is not just elegant theory, but is often the only hope for obtaining an algorithm that has guarantees of both correctness and a running time that terminates within the remaining thermodynamic lifetime of the universe.

At this point, an observer might reasonably ask: if the goal is to design provably efficient algorithms, why insist on a running time \textit{multiplicative} in $n$ of the form $n^{c}\cdot f(k)$? Would the seemingly stronger bound $n^{c}+f(k)$ be even better, especially when $n$ is expected to dwarf $k$? This instinct is sound, but in the classical framing of parameterized complexity, the distinction barely matters: the additive and multiplicative notions of FPT coincide up to polynomial factors\footnote{A problem admits a $\mathrm{poly}(n)+ f(k)$ algorithm if and only if it admits a $\mathrm{poly}(n)\cdot f(k)$ algorithm.} in $n$.

However, today inputs can be so large that even a \textit{single} superlinear pass over one is infeasible. Applications to giant instances were the subject of some discussion at the recent Dagstuhl seminar ``Recent Trends in Graph Decomposition''~\cite{DagstuhlReports13_8}), where the present paper's foundations were conceived.
When linear-time processing is the limit of feasibility, there is a natural question that classical parameterized complexity did not originally need to confront:

\begin{quote}
\textit{Forced to incur only \textbf{linear} dependence on~$n$, do the additive and multiplicative definitions of FPT remain equivalent?}
\end{quote}

Our answer is a definitive \textbf{no}: we draw a sharp distinction between \textit{linear FPT} (\textsc{LFPT})\footnote{Sometimes referred to as fixed-parameter linear (FPL) \cite{KAMMER2015126}.} -- algorithms running in time $n \cdot f(k)$ -- and \textit{truly linear FPT} (\textsc{TLFPT}) -- algorithms running in time $O(n) + f(k)$. Furthermore, we show that TLFPT is a strict subset of LFPT, by diagonalization, yielding the refined hierarchy
\[
    \textsc{TLFPT} \subset \textsc{LFPT} \subset \textsc{FPT} \subseteq W[1] \subseteq W[2] \subseteq \dots
\]

Populating \textsc{TLFPT} requires careful consideration of the assumed model of computation (cf.~\cite{HHRTT24}). 
As is standard in many fine-grained contexts, we work in a word-RAM model. Furthermore, mirroring the classical \textsc{FPT}/kernelization equivalence, we show that being in \textsc{TLFPT} corresponds exactly to admitting kernelization running in time $O(n+k)$\footnote{As we shall see, any problem for which a kernel is computable in time $O(n+f(k))$ for some $f:\mathbb N \to \mathbb N$ also admits a truly linear-time kernelization, i.e. one running in time $O(n+k)$.}.
This insight lets us identify problems that are already known to be in TLFPT (e.g. \textsc{Max Leaf Spanning Tree} on weighted graphs \cite{Jansen2010} and \textsc{$d$-Hitting Set} on hypergraphs \cite{Bevern2014}). 
Linear time kernelization\footnote{In the literature, ``linear time'' has varyingly been used to refer to running times of form $O(n \cdot f(k))$ where $f$ is linear \cite{Mertzios2020}, polynomial but nonlinear \cite{GiannopoulouMN17} (there called ``polynomial-linear FPT''), or superpolynomial \cite{bodlaender_linear-time_1996}. In the formalism of \cite{downey2013fundamentals}, ``kernelization'' entails at most polynomial dependency on $k$.} has also been applied to problems in P. For example, \textsc{Maximum Matching} has been shown to have such a kernelization \cite{Mertzios2020} and this result is bolstered by experimental evidence of its applicability \cite{koana2021data}.

Here, we identify further problems that inhabit \textsc{TLFPT} by establishing appropriate meta-theorems.
We summarize the problems, parameters and methods we use below. These examples illustrate that truly linear FPT

is not prohibitively restrictive,

but a robust and attainable algorithmic goal.

\[
\begin{tabular}{l|l|l}
\textbf{Problems} & \textbf{Parameter} & \textbf{Meta-theorem applied} \\ \hline
\textsc{Vertex Cover} & size of the vertex cover & \Cref{thm:deg-alg} (also~\cite{downey2013fundamentals})\\
\textsc{Dominating Set} & vertex cover number & \Cref{thm:matchTwin}\\
\textsc{Min--Max Matching} & size of the matching  & \Cref{thm:matchTwin}\\
\textsc{Diverse Pair of Matchings} & $|M_{1} \Delta M_{2}|$  & \Cref{thm:matchTwin}\\
 \textsc{$(n-k)$-Coloring} & $k$  & \Cref{thm:matchTwin} \\
\textsc{SAT} & incidence treedepth  & \Cref{thm:matchTwinSubgraphs}\\
\textsc{$k$-Path} & $k$  & \Cref{thm:matchTwinSubgraphs}\\
\textsc{$k$-Vertex Ranking (Treedepth)} & $k$ & \Cref{thm:matchTwinSubgraphs} \\
\textsc{$H$-Coloring} & BFS-width  & \Cref{thm:bfs-nice-gives-tlfpt}\\
\end{tabular}
\]

\begin{figure}
    \centering
 \begin{tikzpicture}[a/.style={draw=blue,fill=white},b/.style={draw=red,fill=white},c/.style={draw=green,fill=white},d/.style={draw,fill=white}
]
\draw (-4,0) -- (7.8,0); 
\clip (-4.1,0) rectangle (7.9,3.8); 
\begin{scope}
    \clip [rounded corners=16mm] (-3,1.6)--(-3,-1.6)--(3,-1.6)--(3,1.6)--cycle;
    \fill [ gray!50!white, rounded corners=30mm] (-0.5,3)--(-0.5,-3)--(6,-3)--(6,3)--cycle;
\end{scope}
\begin{scope}
    \clip [rounded corners=20mm] (0.5,2)--(0.5,-2)--(5,-2)--(5,2)--cycle;
    \fill [gray!50!white, rounded corners=25mm] (-4,2.5)--(-4,-2.5)--(4,-2.5)--(4,2.5)--cycle;
\end{scope}
\draw[rounded corners=3mm, fill=gray!50!white] (-3.5,0.3)--(-3.5,-0.3)--(7.8,-0.3)--(7.8,0.3)--cycle; 
\draw[rounded corners=7mm] (-2,0.7)--(-2,-0.7)--(2,-0.7)--(2,0.7)--cycle; 
\draw[rounded corners=16mm] (-3,1.6)--(-3,-1.6)--(3,-1.6)--(3,1.6)--cycle; 
\draw[rounded corners=25mm] (-4,2.5)--(-4,-2.5)--(4,-2.5)--(4,2.5)--cycle; 
\draw[rounded corners=30mm] (-0.5,3)--(-0.5,-3)--(6,-3)--(6,3)--cycle; 
\draw[rounded corners=20mm] (0.5,2)--(0.5,-2)--(5,-2)--(5,2)--cycle; 
\node[] at (0,0.9){$d$-degree reducible};
\node[] at (-0.3,1.8){$\tau$-twin reducible};
\node[] at (-0.5,2.7){twin subgraph reducible};
\node[] at (3.5,3.2){$\ell$-bounded depth};
\node[] at (3.4,2.2){$m$-matching bounded};
\node[] at (6.8,0.5){BFS-nice};
\draw (-4,0) -- (7.8,0); 
\end{tikzpicture}
\caption{A graphical summary of our results. Containment of the properties we will use to categorize problems within in this paper. The shaded areas show TLFPT results.}
    \label{fig:venn-diagram}
\end{figure}
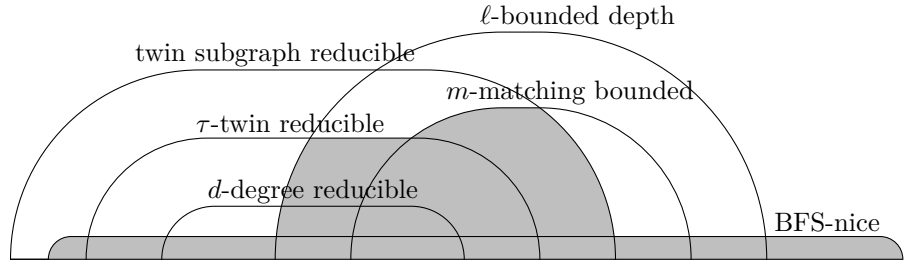

\subparagraph{TLFPT is useful: Big Data and Vertex Cover}

We continue in the tradition of illustrating the value of FPT with a concrete example using \textsc{Vertex Cover}, which asks, given a graph $G=(V,E)$ and integer $k$, whether there exist $k$ vertices in the graph which together cover all edges \cite{Cyg15,downey2013fundamentals,DF99Book}.
In those textbooks' motivating examples, values of $n$ ranged from $1000$ to $10000$, with of $k$ between $10$ and $70$.  
In 2026, major scientific endeavors are supported by the processing of petabytes of data~\cite{CERN:2931994}: for our example, we take $n$ equal to a hundred billion and $k=100$. 
The state-of-the art algorithm for \textsc{Vertex Cover} \cite{CKX10} runs in TLFPT time\footnote{The runtime of the algorithm as presented in \cite{CKX10} is $O(1.2738^k + kn)$, but $n$ denotes the number of vertices; in the definition of TLFPT, $n$ denotes the input size (i.e.~$|V|+|E|$, for a graph). The term $kn$ in the expression comes entirely from Buss kernelization, which takes truly linear time $O(|V|+|E|)$ \cite{downey2013fundamentals}.}: $O(1.27^k + |V| + |E|)$.
Our example values of $n$ and $k$ make stark the distinction between TLFPT and LFPT: an algorithm running in time $1.2738^{k}\cdot n$ would take $10^{21}$ operations to terminate; an algorithm running in time $1.2738^{k} + n$ would take only $10^{11}$ operations to terminate. For an average laptop, this is the difference between minutes and millennia.

\subparagraph{Technical preliminaries}

Let $[n] := \{1,2,\ldots,n\}$ and $[i,j] := \{i,i+1,\ldots,j\}$. Unless stated otherwise graphs $G = (V, E)$ are simple, loopless and undirected. The \textit{union} of graphs $(V,E_1)$ and $(V,E_2)$ is $(V,E_1 \cup E_2)$. 
We refer the reader to Garey and Johnson~\cite{garey2002computers} for the standard definitions of a maximal matching and dominating set, and the corresponding decision problems.
Given a loopless graph $H$, an \textit{$H$-coloring} of a graph $G$ is a graph homomorphism $G \to H$. When $H$ is the complete graph on $n$ vertices, we speak of a \textit{proper $n$-coloring.} The $H$-coloring problem asks if a given graph $G$ admits an $H$-coloring. For convenience, we include the definition of a parameterized decision problem below (see~\cite{DF99Book} for background in parameterized complexity).
\begin{defn}[(Parameterized) decision problem]\label{def:problem}
    
    A \emph{parameterized} problem is a language $L \subseteq \Sigma^* \times \mathbb N$. Given an instance $I=(x,k)$, we call $k$ the \emph{parameter}. The \emph{$k$th slice} $\Pi_k$ of a parameterized problem $\Pi$ is a language $L\subseteq \Sigma^*$ with the property that $(x,k) \in \Pi \iff x \in \Pi_k$.
\end{defn}

\subparagraph{Related work} 
Many works in LFPT are directly relevant to the problems considered here~\cite{Bevern2014,FGJPS24,koana2021data,Mertzios2020,nadara_et_al:LIPIcs.ESA.2022.79,Rei14};
further examples include \cite{Jansen2010,GiannopoulouMN17,HMNN24,LokshtanovEtc18}. 
Many problems that cannot be solved in linear time can be solved in \textit{almost linear} time, for example, in $O(n \cdot \mathit{polylog}(n))$. This is not only true of {\sc Sorting} but also, e.g., for deciding first-order properties in graph classes with locally bounded expansion~\cite{DvorakEtc13}. 
Note that in our context, since $\log(n)$ dominates $k$, $f(k) + n \cdot\mathit{polylog}(n)$ and $f(k) \cdot n \cdot \mathit{polylog}(n)$ define the same class.

\section{Technical Foundations for TLFPT}
We start our technical contributions by defining a class of algorithms and a related class of problems, and some important foundational notions for dealing with these classes.  

\begin{defn}[TLFPT]\label{def:tlfpt}
    An algorithm is \emph{Truly Linear Fixed Parameter Tractable} (TLFPT) if, on input $(I,k)$ with $|I|=n$, it runs in time $O(n)+f(k)$. TLFPT denotes the class of all parameterized problems admitting a TLFPT algorithm.
\end{defn}

\begin{defn}[TLFPT-reduction]
    Let $\Pi$ and $\Gamma$ be two parameterized decision problems, a map $r: \Pi \to \Gamma$ is a \emph{TLFPT-reduction} if there exists a function $f$ such that:
    \begin{itemize}
        \item $r(x,k)$ terminates in time $O(|x|+f(k))$ for each $(x,k)$ in the set of instances of $\Pi$, and returns an instance $(x',k')$ of $\Gamma$,
        \item $(x,k) \in \mathrm{YES}(\Pi) \iff (x',k') \in \mathrm{YES}(\Gamma)$, and
        \item $k' \leq f(k)$.
    \end{itemize}
\end{defn}
    We say that $\Pi$ is TLFPT-reducible to $\Gamma$ if there exists a TLFPT-reduction from $\Pi$ to $\Gamma$.

    In the next lemma, we show that TLFPT-reductions are transitive, and that reducing to a destination problem that is in TLFPT implies that that source problem is also in TLFPT. 
\begin{lemmarep}\label{lem:tlfpt-reducibility-transitive}
    TLPFT-reductions are transitive. In addition, if parameterized problem $\Pi$ is TLFPT-reducible to parameterized problem $\Gamma$ and $\Gamma \in $ TLFPT, then $\Pi \in $ TLFPT. 
\end{lemmarep}
\begin{proof}
First we show transitivity, suppose a parameterized problem $\Pi$ is TLPFT-reducible to parameterized problem $\Gamma$ and $\Gamma$ is TLFPT-reducible to parameterized problem $\Sigma$; we will show that $\Pi$ is TLFPT-reducible to $\Sigma$. Let $r$ be a TLFPT-reduction from $\Pi$ to $\Gamma$ and $q$ be a TLFPT-reduction from $\Gamma$ to $\Sigma$. Then we define a function $s:\Pi \to \Sigma$ by $s((x,k))=q(r((x,k)))$ and show that our conditions hold. 
    Let $f$ be a function such that:
    \begin{itemize}
        \item $r((x,k))$ terminates in time $O(|x|+f(k))$ and has output $(x',k')$ satisfying $k' \le f(k)$, and
        \item $q((x',k'))$ terminates in time $O(|x'|+f(k'))$ and has output $(x'',k'')$  satisfying $k'' \le f(k')$
    \end{itemize}
    
    Then $s((x,k))=q(r((x,k)))$ terminates in time $O(|x|+f(k)+f(f(k))$ and has output $(x'',k'')$ satisfying $k'' \le f(k') \le f(f(k))$. Also, $(x'',k'')\in \mathrm{YES}(\Sigma) \iff (x',k') \in \mathrm{YES}(\Gamma) \iff (x,k) \in \mathrm{YES}(\Pi)$ and thus $s$ is TLFPT-reduction from $\Pi$ to $\Sigma$.

Now we show that a TLFPT-reduction from $\Pi$ to $\Gamma \in$ TLFPT gives us $\Pi \in$ TLFPT. 
    Let $r$ be a TLFPT-reduction from $\Pi$ to $\Gamma$ and $\mathcal{A}$ be a TLFPT algorithm solving $\Gamma$. Let $f$ be a function such that both $\mathcal{A}$ and $r$ run in time $O(n+f(k))$ and the instance $(x',k')$ returned by $r((x,k))$ satisfies $k' \le f(k)$.
    
    For an instance $(x,k)$ of $\Pi$, it is possible to compute $r((x,k))=(x',k')$ and then to compute $\mathcal{A}(x',k')$. The runtime of $r((x,k))$ is clearly TLFPT in the size of the input instance; the runtime of $\mathcal{A}(x',k')$ is TLFPT in $(x',k')$ and therefore is $O(|x'|+f(k')) \le O\left(f(k)+|x|+f(f(k))\right)$ which is also TLFPT. 
\end{proof}

\subparagraph{The word-RAM architecture}\label{sec:why-word-ram}
Results in the truly linear paradigm are highly sensitive to the choice of model.
We use the word-RAM model \cite{FW90-STOC,FW90-FOCS}, as it is the canonical model for the analysis and design of fine-grained algorithms \cite{EvdHM24, HHRTT24}.
More limited models of computation in this context are neither as widely studied nor as practically relevant, since they cannot necessarily support even depth-first search in linear time.
Briefly, word-RAM is a unit-cost RAM with the additional property that the maximum number of bits stored in any single register -- called the \textit{word size} -- is $\log n$, where $n$ is length of the input. 
Thus word-RAM prohibits unbounded parallelism wherein constant-time operations are applied to words of length polynomial in the size of the input. 

On the other hand, ``word-level parallelism'' may be exploited \cite{Hag98}.

Such parallelism includes, for example, taking the bitwise AND of two bitstrings each of length $\ell$ in time $O(1)$ rather than time $O(\ell)$. A more powerful example: let $f: \{0,1\}^{\ell_1} \to \{0,1\}^{\ell_2}$ be some function stored as a lookup table with $2^{\ell_1}$ entries, with $\ell_1, \ell_2$ at most than the word size of our RAM. Then for any $x \in \{0,1\}^{\ell_1}$ we may retrieve $f(x)$ in $O(1)$ time.

As is standard in parameterized complexity, we can assume that $n$ is arbitrarily larger than $k$, and in particular that $\log n > f(k)$ for any function $f$ (this statement is made rigorous in the appendix).
This allows us, for example, to represent a graph of size $k$ with a single word, and operate on it in constant time using bitwise operations.\\

\begin{toappendix}

\begin{lemma}\label{lem:words-big-as-fn-of-k}
    When considering the (non-)existence of a word-RAM algorithm for a parameterized problem $\Pi$, it suffices to fix any function f and consider only those instances $(x,k)$ of $\Pi$ satisfying that $|x| > 2^{f(k)}$. 
\end{lemma}
\begin{proof}
    We shall show that for any function $f$, a word-RAM algorithm $\mathcal A$ deciding problem $\Pi$ in TLFPT time exists if and only if there exists a RAM algorithm $\mathcal B$ with $\max (f(k), \log n)$-bit registers deciding $\Pi$ in TLFPT time. The statement of the lemma follows directly from this.
    ($\impliedby$): for any input $(x,k)$ of the algorithm, we distinguish between two cases:
    \begin{itemize}
        \item if $f(k) \le \log n$, then $\mathcal B$ can directly be run on a word-RAM with no changes to its behavior or runtime.
        \item if $f(k) > \log n$, then $\mathcal B$ can be simulated in a word-RAM (possibly much more slowly). The size of $n$ is bounded by a function of $k$ (specifically, $n < 2^{f(k)}$), so the runtime of the simulation is also bounded by a function of $k$. 
    \end{itemize}
    Note that this does not require that $f$ is computable; simply attempt to run $\mathcal B$ in the word-RAM model, and if ever a register of length greater than $\log n$ is needed then this witnesses that $f(k) > \log n$.\\
    ($\implies$): trivial direction, since $\mathcal A$ itself already decides $\Pi$ using $\log n$-bit registers (which can of course be simulated by larger registers).
\end{proof}

Note that \cref{lem:words-big-as-fn-of-k} entails that a word size of at least $\max(f(k),\log |x|)$ can be assumed.

\end{toappendix}

\subparagraph{TLFPT is Strictly Contained in LFPT}

Our proof mirrors that given in \cite{CR73} by Cook and Reckhow to prove the time hierarchy theorem for RAMs. Our task is to make the complexity two-dimensional (i.e.~in both $n$ and $k$) and tailor the discussion to word-RAMs.

Suppose we have an encoding $\alpha$ which maps every string over $\{1,2\}$ to some word-RAM.
We say that the word $w$ \textit{encodes} the word-RAM $P_w$ which it maps to under $\alpha$. 
We also assume that for each word-RAM $M$, there are infinitely many strings encoding $M$. 
In particular, it is possible to pad $w$ to increase its length without affecting the word-RAM $P_w$ that it encodes.
Let $f\in \omega(1)$ be a computable function\footnote{In \cite{CR73} the function needed to be time-constructible. Here we may mildly relax this assumption.}. We define the following language $A=A(f)$.
\[ \mathbf{A}=\{ (w,k) : P_w \text{ with input } w \text{ halts within time } f(k)\cdot n \text{ without accepting } w\}\]

\begin{lemmarep}
   For any\footnote{We do not need to impose any restriction on $g$.} 
   $\mathbb{N} \xrightarrow{g} \mathbb{N}$ and $c \in \mathbb N$, no word-RAM recognized $\mathbf{A}$ in time $g(k)+c \cdot n$.
\end{lemmarep}
\begin{proof}
We follow the proof of Theorem 3 in \cite{CR73}.
Suppose a word-RAM $P'$ accepts $\mathbf{A}$ in time $g(k)+c\cdot n$. Then there exists some $(w,k)$ (possibly with $|w|$ extremely large) so that $P_w=P'$ and $g(k)+c\cdot |w|<f(k) \cdot |w|$. Now:
\[
    \begin{array}{ll}
        (w,k) \in \mathbf{A} & \rightarrow \mbox{ $P'$ accepts $w$ within time $g(k)+c \cdot |w|$} \\
        & \rightarrow \mbox{ $P_w$ accepts $w$ within time $g(k)+c\cdot |w|$} \\
        & \rightarrow \mbox{ $P_w$ accepts $w$ within time $f(k) \cdot |w|$} \\
        & \rightarrow \mbox{ $w \notin \mathbf{A}$}
    \end{array}
\]
    And:
\[
    \begin{array}{ll}
        (w,k) \notin \mathbf{A} & \rightarrow \mbox{ $P'$ does not accept $w$ within time $g(k)+c\cdot |w|$} \\
        & \rightarrow \mbox{ $P_w$ does not accept $w$ within time $g(k)+c\cdot |w|$} \\
        & \rightarrow \mbox{ $P_w$ does not accept $w$ within time $f(k) \cdot |w|$} \\
        & \rightarrow \mbox{ $w \in \mathbf{A}$}
    \end{array}
\]
\end{proof}
We note the existence of a form of universal RAM in \cite{CR73} called a \textit{Random Access Stored Program} (RASP) machine. A full discussion of its programming language RAM-ALGOL is beyond the scope of this paper. However, In Theorem 1 of \cite{CR73} it is explained how the RASP can simulate another RAM with time overhead of only a constant factor. We may additionally note that the contents of the registers in the RASP need not exceed in length the contents of the registers of the RAM being simulated. It follows that the RASP functions as a universal word-RAM. We note that the RASP can count the $f(k)\cdot n$ steps of the computation in two registers since $f(k)\leq n$ (applying a fortiori the assumption that $f(k) < \log n$ as discussed above). 
We are now in a position to obtain the following (and this uses that $f$ is computable, with time complexity $h$).

\begin{lemma}
    There is a word-RAM that accepts $\mathbf{A}$ in time $h(k)\cdot n$ for some computable function $h$.
\end{lemma}

\begin{theorem}
    $\mathbf{A}\subset$ LFPT $\setminus$ TLFPT.
\end{theorem}

\subparagraph{Linear Time Kernelization}

Kernelization has seen extensive use within parameterized algorithmics, and is central to many of our results. Here we define linear time kernelization and then 
argue that it is fundamentally coupled with the class of TLFPT problems.

We say that a parameterized problem $\Pi$ admits a \textit{linear time kernelization} if there exists a function $g$ and a self-reduction $\pi \colon (x,k) \mapsto (x',k')$, computable in time linear in $|x|+k$, such that: \begin{enumerate*}
    \item $(x,k)$ is a yes-instance if and only if $(x',k')$ is a yes-instance; and
    \item $|x'| \leq g(k)$ and $k' \le k$.
\end{enumerate*} 
If there is a self-reduction respecting properties (1) and (2) but with runtime $O(|x|+f(k))$ for some superlinear $f$ then we refer to ``TLFPT time kernelization'' or say that a kernel for $\Pi$ is computable in time $O(|x|+f(k))$. There is some variability in the definition of kernelization. By default (e.g., in \cite{downey2013fundamentals}), a kernelization's runtime is polynomial in $|x|+k$. In this sense, TLFPT time kernelizations are not strictly kernelizations themselves, though as we shall see in \cref{cor:kern-iff-tlfpt-reduction-iff-tlfpt}, any problem admitting a TLFPT time kernelization also admits a linear-time kernelization.

\begin{theoremrep}\label{thm:lintimekern-iff-tlfpt}
A nontrivial parameterized problem $\Pi$ admits a linear time kernelization if and only if it is in TLFPT.
\end{theoremrep}
\begin{proof}
    
    A kernelization of an instance $(x, k)$ of parameterized problem $\Pi$, linear in $|x|+k$, runs in time  $c|x|+ck$ for some constant $c$ and yields a kernel of size at most some function $g(k)$. Overall, the time complexity to solve $(x,k)$ is $c|x|+ck+g'(k)$ for some $g'$: applying the kernelization takes time $c|x|+ck$ and solving the resulting kernel takes time $f(|x'|,k')$ for some $f$.

    Conversely, we choose fixed yes- and no-instances to reduce to (they exist, since $\Pi$ is nontrivial). There must exist an algorithm $\mathcal A$ with complexity $c|x| + f(k)$ which witnesses TLFPT membership. We distinguish between two cases. 
    \begin{itemize}
        \item If $|x| > f(k)$ then the runtime of $\mathcal{A}((x,k))$ is $c|x| + f(k) \le (c+1)|x|$ which is linear in $|x|$ (and in $|x|+k$), so we can simply solve the instance and then map it to a constant-size yes- or no-instance, as appropriate. 
        \item Otherwise, $|x|\le f(k)$ and we map $(x,k)$ to itself, which satisfies the definition of a linear time kernelization (by letting $g=f$).
    \end{itemize}
    We note that this does not require that $f$ be computable: we can check whether $|x|> f(k)$ by simply running $\mathcal{A}((x,k))$ for $(c+1)|x|$ steps. If it has terminated then we are in the first case (and return a constant-size yes- or no-instance), and otherwise we are in the second case (and return $(x,k)$ directly).     
\end{proof}

\begin{corollary}\label{cor:kern-iff-tlfpt-reduction-iff-tlfpt}
    For any parameterized problem $\Pi$, the following properties are equivalent:
    \begin{enumerate}
        \item a kernel for $\Pi$ is computable in time $O(|x|+k)$, i.e. $\Pi$ admits a linear-time kernelization,
        \item a kernel for $\Pi$ is computable in time $O(|x|+f(k))$ for some function $f$, i.e. $\Pi$ admits a TLFPT time kernelization,
        \item $\Pi$ is decidable in time $O(|x|+f(k))$, i.e. $\Pi$ is in TLFPT.
    \end{enumerate}
\end{corollary}

\begin{proof}
    The implications (1) $\implies$ (2) $\implies$ (3) are trivial, and the implication (3) $\implies$ (1) follows directly from \cref{thm:lintimekern-iff-tlfpt}.
\end{proof}

\section{TLFPT via Degree-Based Kernelization: a Warmup}\label{sec:degree-based-kern}

In this section we describe a general approach to proving membership in TLFPT. In particular, we show that if a problem is \textit{degree reducible} and has \textit{bounded matching number} then we can construct a kernel for the problem in TLFPT time. 

We illustrate this method using \textsc{Vertex Cover}, which is both one of the classic NP-complete combinatorial problems \cite{Kar72}, and also (with parameter $k$) one of the prototypical FPT problems \cite{DF99Book}. 
A linear kernelization algorithm for this problem has been previously described \cite{buss1993nondeterminism} and is known to be linear \cite{downey2013fundamentals}. Here we use the problem to demonstrate our more general method.

\begin{toappendix}
\begin{framed}
    \noindent
    \textsc{Vertex Cover}\\
    \textbf{Input:} graph $G=(V,E)$, integer $k$.\\
    \textbf{Question:} does there exist a set $S\subseteq V$ with $|S|\le k$ such that for every $(u,v) \in E$, $\{u,v\} \cap S \ne \emptyset$?
\end{framed}
 \end{toappendix}

A common observation

is that any vertex with degree at least $\delta$ must be contained in any vertex cover of size at most $\delta - 1$. That is, if a graph $G$ contains some vertex cover $C$ of size at most $k$, then $C$ necessarily contains every vertex of $G$ with degree at least $k+1$. Thus, for any vertex $v$ with degree at least $k+1$, $(G,k)$ is a yes-instance of {\sc Vertex Cover} if and only if the same holds for $(G-v,k-1)$. Further, if $G$ contains some isolated vertex $v$, then $(G,k)$ is a yes-instance of {\sc Vertex Cover} if and only if the same holds for $(G-v,k)$.
We call a problem with these properties \textit{degree reducible} as we formally define below. 

\begin{defn}\label{def:d-degree-reducible}
Let $\Pi$ be some parameterized graph problem. We say that $\Pi$ is $d$-degree reducible for some function $d: \mathbb{N} \to \mathbb{N}$, if for any instance $(G,k)$ of $\Pi$ and any vertex $v$ of $G$ the following holds:
\begin{itemize}
    \item if $|N(v)| > d(k)$, then $(G,k) \in YES(\Pi)$ if and only if $(G-v,k-1) \in YES(\Pi)$; and
    \item if $|N(v)| =0$, $(G,k) \in YES(\Pi)$ if and only if $(G-v,k) \in YES(\Pi)$.
\end{itemize}
\end{defn}

As suggested above, \textsc{Vertex Cover} is $d$-degree reducible where $d$ is the identity function.

\begin{lemmarep}\label{lem:vc-deg}
    \textsc{Vertex Cover} is $d$-degree reducible when $d(k)=k$ for each $k \in \mathbb{N}$.
\end{lemmarep}
\begin{proof}
    We show $d$-degree reducibility of \textsc{Vertex Cover} by showing each condition in Definition~\ref{def:d-degree-reducible} holds.

    Let $(G,k)$ be an instance of \textsc{Vertex Cover}. Suppose $G$ contains some vertex $v$ with degree at least $k+1$. We claim that $G$ contains a vertex cover of size at most $k$, if and only if $G-v$ contains a vertex cover of size at most $k-1$.

    Suppose that there is some vertex cover $C \subseteq V$ of size at most $k$. We claim that $v \in C$. Otherwise, given every edge in $G$ has at least one endpoint in $C$, $N(v) \subset C$. This is a contradiction as $|N(v)| > |C|$. It follows that $C\setminus \{v\}$ is a vertex cover of $G-v$ and so $(G-v,k-1)$ must also be a yes-instance of \textsc{Vertex Cover}.

    Now let $(G-v,k-1)$ be a yes-instance of \textsc{Vertex Cover}, that is, there is some vertex cover $C \subseteq V \setminus \{v\}$ of size at most $k-1$. It follows that $C \cup \{v\}$ is a vertex cover of $G$ with size at most $k$.
    
    We now consider isolated vertices. Let $v$ be some vertex with degree $0$ in $G$. It follows that $v$ is not contained in any minimal vertex cover of $G$ and so $(G,k)$ is a yes-instance of \textsc{Vertex Cover} if and only if $(G-v,k)$ is.
\end{proof}

The following now holds by definition:

\begin{obs}\label{obs:remove-high-deg}
    Let $\Pi$ be some $d$-degree reducible problem and let $G$ be a graph. For any set $R$ of vertices with degree greater than $d(k)$ and any set $I$ of vertices which are isolated in $G-R$, $(G,k) \in YES(\Pi)$ if and only if $(G-(R \cup I),k-|R|) \in YES(\Pi)$.
\end{obs}

The first phase of our algorithm will, in time $O(|E|+|V|+d(k))$, compute the set $R$ of vertices of degree greater than $d(k)$ and the set $I$ which are isolated in $G-R$. We then build the graph $G-(R \cup I)$. By definition, this graph has degree at most $d(k)$ and contains no isolated vertices.

We then make use of a second property, both of \textsc{Vertex Cover}, as well as many other problems regarding matching. If some graph $G$ contains a matching of size $m$, then necessarily every vertex cover of $G$ has size at least $m$. It follows that if $G$ contains a matching of size $k+1$ then $(G,k)$ is a no-instance of \textsc{Vertex Cover}. More generally, we consider problems which become trivial if our graph contains a sufficiently large matching. We say that such problems have \textit{bounded matching number}, as defined below.

\begin{defn}\label{def:m-bounded-matching}
    Let $\Pi$ be some parameterized graph problem. We say that $\Pi$ has \emph{positive (respectively negative) $m$-bounded matching number} for some function $m: \mathbb{N} \to \mathbb{N}$, if, given any instance $(G,k)$ of $\Pi$, whenever $G$ contains a matching of size greater than $m(k)$, then $(G,k)$ is a positive (resp.~negative) instance of $\Pi$. Note that this trivially holds if no instance of $\Pi$ contains a matching of size greater than $m(k)$.
\end{defn}

Thus, \textsc{Vertex Cover} has negative $m$-bounded matching number:

\begin{lemmarep}\label{lem:vc-match}
    \textsc{Vertex Cover} has $m$-bounded matching number for $m(k)=k$.
\end{lemmarep}
\begin{proof}
    Let $(G,k)$ be an instance of {\sc Vertex Cover}. If $G$ contains a matching of size at least $k+1$, then any vertex cover of $G$ has size at least $k+1$. It follows that $(G,k)$ is a no-instance and {\sc Vertex Cover} has $m(k)$-bounded matching number, where $m$ is the identity function.
\end{proof}

Our later proofs shall make use of the functions $m$ and $d$ being monotone increasing without loss of generality:

\begin{lemma}\label{lem:monotone-increasing-funcs-match}
    If a problem $\Pi$ has $m$-bounded matching number (respectively is $d$-degree reducible) then there exists a function $m'$ (resp.~$d'$) which is monotone increasing such that $\Pi$ has $m'$-bounded matching number (resp.~is $d'$-degree reducible).
\end{lemma}
\begin{proof}
    Let $f:\mathbb N \to \mathbb N$ be any function. Define $f'(x) = \max(\{f(i) + x: 1 \le i \le x\})$. It is clear that $f'$ is monotone increasing, and additionally that if a problem has $f$-bounded matching number (respectively is $f$-degree reducible) then it also has $f'$-bounded matching number (resp.~is $f'$-degree reducible).
    
    This follows from the fact that $f'(x) \ge f(x)$ for all $x$, and that the usage of the function $f$ in the respective definitions is to give a bound beyond which distinction is irrelevant (i.e. graphs containing matchings larger than $m(k)$ are either all YES-instances or all NO-instances, vertices of degree greater than $d(k)$ can all be reduced away).
\end{proof}

We now leverage the properties of $d$-degree reducibility and $m$-bounded matching number to kernelize our problems.

\begin{observation}\label{obs:VC-kernel-size}
    Let $G$ be a graph with maximum degree $\Delta$ and no isolated vertices and let $M$ be some maximum matching, then $G$ has size at most $2|V(M)|\Delta$.
\end{observation}

We highlight that for any graph $G$ we can find some arbitrary maximal matching $M$, in time $O(|V|)$ using depth first search \cite{CLRS}. This and reduction of high degree vertices allows us to obtain a kernel of size $2k^2$ for {\sc Vertex cover} in time $O(n)+f(k)$. We first summarize our algorithm for {\sc Vertex cover}, before giving the more general algorithm in full detail alongside a formal analysis of runtime.

\pagebreak

\textbf{Algorithm} VertexCover$(G,k)$
\begin{enumerate}
    \item Find the set $R$ of vertices each with degree greater than $k$ in $G$.
    \item Find the set $I$ of vertices with degree $0$ in $G-R$.
    \item Construct the graph $\hat{G} = G-(R \cup I)$. (note that $\hat{G}$ may be disconnected)
    \item Let $\hat{k} = k - |I|$.
    \item Find some arbitrary maximal matching $M$ of $\hat{G}$.
    \item If $|M| > k$, then return $(M_{k+1}, \hat{k})$ where $M_{k+1}$ is that graph consisting exactly of a matching of size $k+1$.
    \item Else, return $(\hat{G}, \hat{k})$.
\end{enumerate}

\begin{algorithm}[!ht]
    \caption{Solving $d$-degree reducible and $m$-matching bounded graph problems.}\label{alg:d-deg-m-match}
    \begin{algorithmic}[1]
        \Require $G=([n],E)$ is a graph on $n$ vertices; integer $k$; computable functions $m$ and $d$.
        \State Compute \texttt{deg} the array of length $n$ with \texttt{deg}[$u$] the degree of vertex $u$. \label{alg-line:deg}
        
        \State Compute $d(k)$. 
        \State Let $G'=(V',E')$ be the graph induced by $V \setminus \{u : \texttt{deg}[u] > d(k)\}$ in $G$ and let $k'=k - |\{u : \texttt{deg}[u] > d(k)\}|$. \label{alg-line:induced-subgraph-round-one}
        \State Recompute \texttt{deg} for the graph $G'$. \label{alg-line:recompute-deg}

        \State Let $G''=(V'',E'')$ be the graph induced by $V' \setminus \{u : \texttt{deg}[u] = 0\}$ in $G$ and let $k''=k'$. \label{alg-line:induced-subgraph-round-two}
        
        \State By depth-first search on $G''$, and greedily build a maximal matching $M \subseteq E(G'')$. \label{alg-line:dfs-on-induced-subgraph}
        \If{$|M|>m(k'')$} 
        \Return $(m(k'')+1)P_2$ (i.e. the matching on $m(k''+1)$ edges)\label{alg-line:case-big-matching}
        \Else 
        ~\Return $G'', k''$ \label{alg-line:case-small-matching}
        \EndIf
    \end{algorithmic}
\end{algorithm}

\begin{theoremrep}\label{thm:deg-alg}
    Given an instance $(G,k)$ of a $d$-degree reducible problem $\Pi$ which has $m$-bounded matching number, there is some algorithm that  runs in $O(|V|+|E|+d(k))$ time and returns a kernel of $\Pi$ on at most $2m(k)\cdot (d(k)+1)$ vertices.
\end{theoremrep}

\begin{proof}
    In particular, Algorithm~\ref{alg:d-deg-m-match} has these properties.
    The algorithm begins (line \ref{alg-line:deg}) by initializing \texttt{deg} as a zero-array of length $|V|$ and then iterating of $E$ and incrementing entries of \texttt{deg} accordingly for each edge. This requires time $O(|E|+|V|)$  (likewise for line \ref{alg-line:recompute-deg} later on). The algorithm then computes $d(k)$.

    In line~\ref{alg-line:induced-subgraph-round-one} (resp.~\ref{alg-line:induced-subgraph-round-two}), the algorithm simply iterates over vertices and edges and uses the lookup table \texttt{deg} (resp.~\texttt{deg}) to decide whether to include them in the newly constructed instance $(G',k')$ (resp.~$(G'',k'')$. This runs in time $O(|V|+|E|)$ and can be implemented in-place, so that the memory used by $G$ is overwritten to hold $G'$ and then $G''$.
    
    Finally, line \ref{alg-line:dfs-on-induced-subgraph} runs in $O(|V|+|E|)$ time as detailed in \cite{CLRS}, (including if $G'$ is disconnected). The algorithm then returns one of two possible instances. Therefore, the algorithm requires a total of $O(|V|+|E|+d(k))$ time.

    Note that applying \cref{def:d-degree-reducible} directly entails that $(G'',k'')$ is a yes-instance iff $(G,k)$ was a yes-instance.
    
    Line \ref{alg-line:case-big-matching} is correct by definition of $m$-matching bounded (Definition~\ref{def:m-bounded-matching}). This follows from the fact that if $|M| > m(k'')$, the graph returned has a matching of size greater than $m(k'')$. The returned graph is on $2+2m(k'')$ vertices and $1+m(k'')$ edges.
    
    Otherwise, $|M| \le m(k')$ and line \ref{alg-line:case-small-matching} returns $(G'',k'')$ with the property that $G''$ is a graph on at most $2m(k)\cdot (d(k)+1)$ vertices. $G''$ admits $M$ as as a maximal matching. We denote $M_V$ the set of vertices incident to an edge of $M$ (so $|M_V| \le 2m(k'')$). Note that since there are no vertices of degree $0$ in $G''$ and $G''$ has maximum degree at most $d(k)$, we obtain that $M_V$ is a dominating set of maximum degree $d(k)$ and thus $G''$ is a graph on at most $2m(k'')\cdot (d(k)+1)$ vertices.
    Since $k'' = k' \le k$ and $m$ and $d$ are monotone increasing, we obtain that $|V(G'')| \le 2m(k) \cdot (d(k)+1)$.
\end{proof}

By Lemmas~\ref{lem:vc-deg} and~\ref{lem:vc-match},

\textsc{Vertex Cover} is $d$-degree bounded and has $m$-bounded matching number, this allows us to apply Theorem~\ref{thm:deg-alg}.
    
\begin{corollary}\label{cor:vc-deg-match}
    \textsc{Vertex Cover} is in TLFPT parameterized by $k$.
\end{corollary}

\section{TLFPT kernelization via Matchings and Twins: the Census Method}\label{sec: Matchings and Twins}

We now describe a second method for linear time kernelization. 

To apply this, the problem must satisfy two properties. The first is \textit{bounded matching number} (c.f.~Definition~\ref{def:m-bounded-matching}) which implies that for any non-trivial instance we can compute a vertex cover of bounded size in linear time. We then construct a \textit{census}, with respect to this vertex cover, which characterizes the remaining \textit{cloud} of vertices according to their neighborhood, thus describing an equivalence relation over the vertices.

Two vertices are \textit{open twins} if they are non-adjacent and have the same neighborhood. That is, a set of open twins must be an independent set. As each equivalence class corresponds to a set of open twins, the second property that we use, \textit{twin reducibility}, allows us to bound the size of each equivalence class and yields a problem kernel.

 To illustrate this method, we show that \textsc{Dominating Set} is in TLFPT when paramterized by vertex cover number. Note that \textsc{Dominating Set} is shown to be W[2]-hard with respect to the size the dominating set~\cite{downey1995fixed}; we instead parameterize by the vertex cover number.
\begin{toappendix}
    \textsc{Dominating Set} is defined as follows.

\begin{framed}
    \noindent
    \textsc{Dominating Set}\\
    \textbf{Input:} graph $G=(V,E)$, integer $c$.\\
    \textbf{Question:} is there a set $S\subseteq V$ with $|S|\le c$ such that for each $u \in V$, $N_G[u]\cap S \ne \emptyset$?
\end{framed}

\end{toappendix}

We use \textsc{DomSetVC} to refer to the problem \textsc{Dominating Set} parameterized by vertex cover number $k$ of the input graph $G$. That is, given an integer $c$ and a graph $G$ with vertex cover number $k$, we ask if there is a dominating set $S$ of $G$ of size at most $c$.

Our kernelization algorithm for \textsc{DomSetVC} begins by, in linear time, removing any isolated vertices in the graph and decrementing $c$ for each of these. Observe that, following this, any vertex cover of cardinality $c$ in $G$ is also a dominating set. This allows us to assume that $c<k$, otherwise we have a trivial yes-instance. This yields the following lemma.

\begin{lemmarep}\label{lem:ds-bounded-matching}
    \textsc{DomSetVC} has negative $m$-bounded matching number where $m(k)=k$.
\end{lemmarep}
\begin{proof}
    Note that the set of yes-instances with maximal matching greater than $k$ is empty.
\end{proof}

\begin{toappendix}
Following from above, we can now assume that we have some maximal matching $M$ of size at most $m(k)$, that is $|V(M)| \leq 2m(k)$. Observe that, since this matching is maximal, every edge of $G$ has at least one endpoint in this matching, hence those vertices not contained in $M$ must form an independent set. We call these vertices the \textit{independence cloud}, $I$.
We now describe how those vertices of $I$ can be represented using a \textit{census vector}. In defining this, we assign every vertex in $I$ a label according to its neighbors in $M$. This is a bit-string of length $2m(k)$ (a bit for each vertex in the maximal matching), such that the $i$th bit of this label is $1$ if the vertex in question is adjacent to vertex $i$ in $M$.

\begin{defn}\label{def:open-twins}
    A subset $S$ of vertices of some graph $G$ is a set of \underline{open twins} if for all $x$ and $y$ in $S$, $N(x) = N(y)$ in $G$. 
\end{defn}

\begin{obs}\label{obs: open twins are independent}
    Any set $S$ of open twins is an independent set. 
\end{obs}

Notice that any two vertices with the same label in the census vector relative to some maximal matching are open twins.

We note that each label forms an equivalence class of open twins in the independence cloud, and there are at most $2^{m(k)}$ such labels. We create a \textit{census vector} which, for each of the equivalence classes, counts how many of the vertices in the independence cloud has the corresponding label.
 
To ensure that the algorithm runs in TLFPT time overall, we truncate the count by a function of $k$. That is, we take the minimum of some $\tau(k)$ and the number of vertices with that label. We refer to problems where we can truncate the count in the census vector by some function $\tau$, as $\tau$-twin reducible.
\end{toappendix}

The second property that we require for the census method ensures that we only have to consider a bounded number of vertices with the same neighborhood. For dominating set, we need only consider at most $c$ vertices with the same neighborhood. This follows from the fact that no dominating set of size $c$ can be a strict subset of a set of at least $c$ open twins; either the remaining vertices are undominated, or all are dominated by a shared neighbor. With our preprocessing of isolated vertices, this implies there is always at least one twin in the set which is not in a dominating set of size $c$. Ensuring that each such vertex is dominated will also dominate all vertices not included in our bounded census.

\begin{defn}\label{def:tau-twin-reducible}
    Let $\Pi$ be a some parameterized graph problem. We say that $\Pi$ is $\tau$\emph{-twin reducible} for some function $\tau: \mathbb{N} \to \mathbb{N}$ if the sets $YES(\Pi_k)$ and $NO(\Pi_k)$ are closed under the deletion of any vertex with at least $\tau(k)$ open twins.
\end{defn}

\begin{toappendix}

\begin{lemma}\label{lem:monotone-increasing-funcs-tau}
    If a problem $\Pi$ is $\tau$-twin reducible then there exists a function $\tau'$ which is monotone increasing such that $\Pi$ is $\tau'$-twin reducible.
\end{lemma}

\begin{proof}
    Let $f:\mathbb N \to \mathbb N$ be any function. Define $f'(x) = \max \{f(i) + x: 1 \le i \le x\}$. It is clear that $f'$ is monotone increasing, and additionally that if a problem is $f$-twin reducible then it is also $f'$-twin reducible (since $f$ is  pointwise at most $f'$).

\end{proof}

\end{toappendix}

\begin{lemmarep}\label{lem:ds-tau-twin}
    \textsc{DomSetVC} is $\tau$-twin reducible for $\tau(k)=k+1$.
\end{lemmarep}
\begin{proof}
    Suppose $S$ is a set of open twins. We will show that, if $|S| > k$, then for any $v \in S$ we have that $(G,c,k)$ is a yes-instance if and only if $(G-v,c,k)$ is.

    We begin by assuming that $D$ is a dominating set of $G$. If $S \cap D = \emptyset$, then $D$ dominates $G-v$. Otherwise, consider some $v \in D \cap S$. Since $|S|> k \geq |D|$, there must be some vertex in $S\setminus D$; call it $u$. Then, since $u$ and $v$ are open twins, $(D\setminus \{v\})\cup \{u\}$ dominates $G-v$.

    Now consider a dominating set $D'$ of $G-v$. Note that, since $|S|> k \geq |D'|$, there must be a vertex $u$ in $S\setminus D$ which is dominated by $D'$. Since $v$ shares a neighborhood with $u$, it must be dominated by the same vertex in $D'$ in $G$. Therefore, $D'$ is a dominating set of size $k$ in $G$, and we have a yes-instance.
\end{proof}

As our problem has matching number at most $m(k)$, it contains a vertex cover with size at most $2m(k)$. It follows that the remaining vertices of the graph can be partitioned into $2^{2m(k)}$ equivalence classes. From twin reducibility, we need only consider $\tau(k)$ vertices from each class, thus yielding a kernel with size only a function of $k$. This gives us the following theorem -- the main algorithmic result of this section.

\begin{theoremrep}[The Census Method]\label{thm:matchTwin}
    For any parameterised graph problem $\Pi$, if there exist functions $\tau: \mathbb{N} \to \mathbb{N}$ and $m: \mathbb{N} \to \mathbb{N}$ such that $\Pi$ is $\tau$-twin reducible and has $m$-bounded matching number, then a kernel for $\Pi$ is computable in time $O(|E(G)|+f(k))$.
\end{theoremrep}
\begin{proof}
    Let $\Pi_k$ be some parameterized graph problem which is both $\tau$-twin reducible and has $m$-bounded matching number. For any instance of $\Pi_k$ with corresponding graph $G = (V,E)$ we describe the following kernelization algorithm running in time $O(|E|) +f(k)$ for some fixed function $f$. Let $M \subseteq E$ be some arbitrary maximal matching of $G$ (note that such a matching can be found in time $O(|V|)$ via depth first search \cite{CLRS}). As $\Pi_k$ has $m$-bounded matching number, if $|M| > m(k)$, then 
    we return the graph $K_{2m(k)+2}$.
        
    Assume now that $|M| \leq m(k)$. We now describe the \textit{matching census} representation of a graph with respect to $M$. Note that this will allow us to find and remove twin vertices in the necessary time bounds. Let $M_V = V(M)$. As $M$ is a maximal matching, $M_V$ is a vertex cover and $G$ can be partitioned into the \textit{matching} graph $G[M_V]$ (which admits $M$ as a perfect matching) and the graph $G[V \setminus M_V]$ consisting of only isolated vertices. From this, we describe the following encoding of $G$. We first fix some arbitrary bijection $f: 2^{|M_V|} \to 2^{M_V}$. Let $\vec{c}$ be that vector of length $2^{|M_V|}$ such that, for every $i \in [2^{|M_V|}]$, there are $\vec{c}[i]$ vertices in $V \setminus M_V$ with neighborhood $f(i)$. We call the tuple $(G[M_V],\vec{c})$ the matching census representation of $G$ with respect to $M$.

    For any $t \geq 1$, we call the array $\vec{c'}$ such that $\vec{c'}[i] = min(t, \vec{c}[i])$ for every $i \in [2^{|M_V|}]$, the $t$-pruning of $\vec{c}$. Further, we call the graph corresponding to $(G[M_V], \vec{c}')$ the $t$-pruning of $G$, with respect to $M$. As $\Pi_k$ is $\tau$-twin reducible, the $\tau(k)$-pruning of $G$ is a kernel for $G$.
    
    We now claim that given such a graph $G$ and maximal matching $M \subseteq G$ with size at most $m(k)$,  we can return the matching census representation of the $\tau(k)$-pruning of $G$ with respect to $M$, in time $O(|E|) + f(k)$ for some fixed function $f$. As $|M| \leq 2m(k)$, the $\tau(k)$-pruning of $G$ with respect to $M$, has size at most $2m+ \tau(k) \cdot 2^{2m(k)}$. Hence, the problem $\Pi_k$ admits a linear time kernelization.

    Algorithm~\ref{alg:census} computes the census vector of the $\tau(k)$-pruning of $G$ with respect to a some maximal matching $M$, where $s = 2|M|$. By Lemma~\ref{lem:words-big-as-fn-of-k}, we can assume that $2^k+1 \leq \log n$, it follows that the word operations in Algorithm~\ref{alg:census} take constant time. We fix an arbitrary bijection $V \to [n]$ such that the vertices in $[s]$ correspond to those contained in $M$. Since $s \le 2m(k)$, the arrays $\vec{c}$ and $\vec{f}$ can be constructed in time $f(k)$, for some fixed function $f$. Moreover, in lines $3$ though $7$, each edge of $G$ is processed at most once. It follows that Algorithm~\ref{alg:census} runs in time $O(|E|) +f(k)$. 
    
    \begin{algorithm}
    \caption{Solving $\tau$-twin reducible and $m$-matching bounded graph problems.}\label{alg:census}
    \begin{algorithmic}[1]
    \Require $G=([n],E)$ is a graph on $n$ vertices; $[s]$ is a vertex cover of $G$; and $\tau(k)$ is some constant.
    \State Allocate $\vec{c}$ to be an array of zeroes of length $2^{s}$.
    \State Let $\vec{f}$ be an array of length $s$ such that $\vec{f}[i]=2^i$ for each $0 \le i < s$. 
    \ForAll{$v \in \{s,s+1,\ldots,n-1\}$}
    \State Let $p=0$. \Comment{Stores target index.}
    \ForAll{$u \in N(v)$}
    \State Let $p+=\vec{f}[u]$.
    \EndFor
    \State Let $\vec{c}\,[p] = min(\tau(k), \vec{c}\,[p]+1)$ 
    \EndFor
    \State \Return $\vec{c}$.
    \end{algorithmic}
    \end{algorithm}
\end{proof}

\begin{toappendix}
 Combining \cref{lem:ds-bounded-matching,lem:ds-tau-twin} and \cref{thm:matchTwin} gives us the following result.
 \begin{thm}\label{thm:ds-tlfpt}
     \textsc{DomSetVC} admits a TLFPT time kernelization.
 \end{thm}
     
\end{toappendix}

In the following subsections we will apply \cref{lem:ds-bounded-matching,lem:ds-tau-twin} and \cref{thm:matchTwin} to a variety of parameterized graph problems, thus showing they are also in TLFPT.

\begin{toappendix}
\subsection{Minimum Maximal Matching}\label{sec:min-max-match}
We continue to show that this method gives inclusion in TLFPT for some other problems. We begin with \textsc{Minimum Maximal Matching}, which is NP-hard~\cite{yannakakis1980edge}.

\begin{framed}
    \noindent
    \textsc{Minimum Maximal Matching}\\
    \textbf{Input:} graph $G=(V,E)$, integer $k$.\\
    \textbf{Question:} does there exist a maximal matching $M\subseteq E$ with size at most $k$?
\end{framed}

\begin{theorem}\label{thm: min max match}
    {\sc Minimum Maximal Matching} admits a TLFPT time kernelization.
\end{theorem}
\begin{proof}
    We first claim that {\sc Minimum Maximal Matching} parameterized by $k$ is $\tau$-twin reducible and has $\ell$-bounded matching number for the functions $\tau(k) = k+1$ and $\ell(k) =2k$.

    Let $(G, k)$ be an instance of {\sc Minimum Maximal Matching}. Suppose $G$ has a maximal matching $M \subseteq E$ of size at most $k$. Let $M'$ be some arbitrary matching in $G$. Since $M$ is maximal, every edge $uv \in M'$ must have at least one endpoint (either $u$ or $v$) that is incident to an edge in $M$. As each edge in $M$ can be incident to at most two vertices, it follows that $|M'| \leq 2|M| \leq 2k$. Thus, if $G$ contains a matching of size at least $2k + 1$, then $G$ cannot have any maximal matching of size at most $k$ and so {\sc Minimum Maximal Matching} has $\ell$-bounded matching number.
    
    We now claim that {\sc Minimum Maximal Matching} is $\tau$-twin reducible. Suppose there exists a set  $S$ of open twins of size $k + 2$ and let $v$ be one of its elements. We claim that $G$ has a maximal matching of size at most $k$ if, and only if, the same holds for $G - v$.

    Let $M$ be some maximal matching in $G$ with size at most $k$. As $G[S ]$ is an independent set (Observation~\ref{obs: open twins are independent}) and $S$ has size $k+2$, there are at least two vertices in $S$ that do not correspond to the endpoint of some edge in $M$. Without loss of generality, one of these is $v$. It follows that $M$ is also a maximal matching in $G-v$.

    Conversely, suppose $M$ is a maximal matching in $G - v$ with size at most $k$. As $G[S]$ is an independent set (Observation~\ref{obs: open twins are independent}) and $S$ has size $k+2$ there exists some vertex $v' \in S$ distinct from $v$ which is not the endpoint of any edge in $M$. As $M$ is a maximal matching, every vertex in $N(v')$ is the endpoint of some edge in $M$. Given that $N(v) = N(v')$, $M$ is also a maximal matching in $G$. Thus, {\sc Minimum Maximal Matching} is $\tau$-twin reducible.

    It follows that applying Theorem~\ref{thm:matchTwin} we obtain a TLFPT time kernelization algorithm for {\sc Minimum Maximal Matching}.
\end{proof}

\subsection{Diverse Pair of Matchings}\label{sec:diverse-match}

We turn to a different matching problem; namely, \textsc{Diverse Pair of Matchings}. For two sets $X$, $Y$ we let $X \triangle Y$ denote the symmetric difference. That is $X \triangle Y = (X \setminus Y) \cup (Y \setminus X)$.

\begin{framed}
    \noindent
    \textsc{Diverse Pair of Matchings}\\
    \textbf{Input:} graph $G=(V,E)$, integer $k$.\\
    \textbf{Question:} does there exist a pair of matchings $M, M' \subseteq E$ such that $|M \triangle M'| \geq k$?
\end{framed}

We note that the ideas necessary to show that {\sc Diverse Pair of Matchings} admits a TLFPT time kernelization were already presented in \cite{FGJPS24}. Moreover, the algorithm described in Theorem~\ref{thm:matchTwin} can be adapted to produce the same $O(k^2)$ kernel in time $O(|V| + k)$. 
For completeness we repeat those ideas of \cite{FGJPS24} as well as giving that TLFPT time kernelization algorithm.

\begin{theorem}\label{thm:diverse-match}
    {\sc Diverse Pair of Matchings} admits a TLFPT time kernelization.
\end{theorem}
\begin{proof}
    We first show that {\sc Diverse Pair of Matchings} is $\tau$-twin reducible and has $\ell$-bounded matching number, for the functions $\tau(k) = 2k$ and $\ell(k) = k-1$. That is, by Theorem~\ref{thm:matchTwin} we obtain a TLFPT time kernelization algorithm for {\sc Diverse Pair of Matchings}.

    Let $(G,k)$ be an instance of {\sc Diverse Pair of Matchings}. If a $G$ contains some matching $M$ of size $k$, then as $M \triangle \emptyset \geq k$, $G$ is a yes-instance of {\sc Diverse Pair of Matchings}. That is, {\sc Diverse Pair of Matchings} has $\ell$-bounded matching number.

    We now show that {\sc Diverse Pair of Matchings} is $\tau$-twin reducible for the function $\tau(k) = 2k$. To show this we repeat the stronger claim which appears in \cite{FGJPS24}. Suppose, for some vertex $v$, there is a set $X \subseteq G$ such that, $v \notin X$, $G[X \cup \{v\}]$ is an independent set and, for every $u \in N(v)$, $|N(u) \cap X| \geq 2k$. Then $(G,k)$ is a yes-instance of  {\sc Diverse Pair of Matchings} if, and only if, $(G-v,k)$ is a yes-instance of  {\sc Diverse Pair of Matchings}. It follows that, if there is some set $S$ of size $2k$ such that $v \notin S$ and, for every $v' \in S$, the vertices $v$ and $v'$ are open twins, then $(G,k)$ is a yes-instance of  {\sc Diverse Pair of Matchings} if, and only if, $(G-v,k)$ is a yes-instance of  {\sc Diverse Pair of Matchings}. That is, {\sc Diverse Pair of Matchings} is $\tau$-twin reducible for $\tau(k) = 2k$.

    It follows, by Theorem~\ref{thm:matchTwin}, that {\sc Diverse Pair of Matchings} admits a TLFPT time kernelization. We note that that kernel resulting from Theorem~\ref{thm:matchTwin} has size $O(2k \cdot 2^{k-1})$, we now claim that we can adapt this algorithm to obtain that kernel of size $O(k^2)$.

    \begin{claim}
        {\sc Diverse Pair of Matchings} admits a kernel of size $O(k^2)$ which can be found in time $O(|E|)+f(k)$, for some fixed function $f$.
    \end{claim}
    \begin{claimproof}
        Our kernelization algorithm for {\sc Diverse Pair of Matchings} will follows very similarly to that of Theorem~\ref{thm:matchTwin} bar a few notable differences which we highlight below.

        Let $(G,k)$ be an instance of {\sc Diverse Pair of Matchings}. Let $M \subseteq E$ be some arbitrary maximal matching of $G$. If $|M| \geq k$, then $G$ is a yes-instance of {\sc Diverse Pair of Matchings}, hence hence we return the graph $K_{2k}$.
        
        Assume now that $|M| \leq k-1$. We will again make use of the matching census representation of a graph and return the matching census representation of our kernel. Note that in the algorithm of Theorem~\ref{thm:matchTwin} we greedily include a vertex $v \in V \setminus V(M)$ in the kernel if there are at most $2k-1$ vertices $v' \in V \setminus V(M)$ such that $v'$ is included in the kernel and $v$, $v'$ are open twins. In this algorithm we will instead include a vertex $v \in V \setminus V(M)$ in the kernel if there exists some vertex $u \in N(v)$ such that there are at most $2k-1$ vertices in $N(u) \cap (V \setminus V(M))$ which are included in the kernel. Algorithm~\ref{alg:diverse-pair} computes the census vector corresponding to those vertices of $V \setminus V(M)$ contained in this kernel. 

        Once again, we let $s = 2|M|$ and fix an arbitrary bijection $V \to [n]$ such that the vertices in $[s]$ correspond to those contained in $M$. Since $s \le 2m(k)$, the arrays $\vec{c}$, $\vec{n}$ and $\vec{f}$ can be constructed in time $f(k)$, for some fixed function $f$. Moreover, in lines $4$ though $14$, each edge of $G$ is processed at most twice, once in lines $7$ through $10$ and once in lines $12$ through $14$. It follows that Algorithm~\ref{alg:diverse-pair} runs in time $O(|E|) +f(k)$.

        \begin{algorithm}
        \caption{Diverse Pair Algorithm}\label{alg:diverse-pair}
        \begin{algorithmic}[1]
        \Require $G=([n],E)$ is a graph on $n$ vertices and vertex cover $M_v = [s]$.
        \State Allocate $\vec{c}$ to be an array of zeroes of length $2^{s}$.
        \State Allocate $\vec{n}$ to be an array of zeroes of length $s$.
        \State Let $\vec{f}$ be a vector of length $s$ such that $\vec{f}[i]=2^i$ for each $0 \le i < s$.

        \ForAll{$v \in \{s,s+1,\ldots,n-1\}$}
            \State Let $p=0$. \Comment{Stores target index.}
            \State Let $f=0$. \Comment{Flag for adding to kernel.}
            \ForAll{$u \in N(v)$}
                \State Let $p+=\vec{f}[u]$.
                \If{$\vec{n}[u] \leq 2k-1$}
                    \State $f=1$
                \EndIf
            \EndFor
            \If{$f=1$}
                \ForAll{$u \in N(v)$}
                    \State $\vec{n}[u] += 1$
                \EndFor
                \State Let $\vec{c}\,[p] += \vec{c}\,[p]+1)$
            \EndIf
        \EndFor
        \State \Return $\vec{c}$.
        \end{algorithmic}
        \end{algorithm}
    \end{claimproof}
\end{proof}

\subsection{\texorpdfstring{$(n-k)$}{(n-k)}-Coloring}
We further show that there exist problems which themselves are not $\tau$-twin reducible nor do they have $m$-bounded matching number but we can reduce them to a problem with these properties. To illustrate this we use the problem $(n-k)$-Coloring. We note that a \textit{proper} coloring of a graph is one where the endpoints of any edge are a colored differently.

\begin{framed}
    \noindent
    \textsc{$(n-k)$-Coloring}\\
    \textbf{Input:} A graph $G=(V,E)$ with $n$ vertices and integer $k$.\\
    \textbf{Question:} Does there exist an $(n-k)$-proper coloring of $G$?
\end{framed}

Although we are simply properly coloring the vertices in the graph, an intuitive way to think about this problem is as if the input graph is trivially colored with $n$ colors (each vertex assigned a different color) and we have to \textit{save} at least $k$ colors by reusing them for multiple vertices; resulting in an $(n-k)$-proper coloring of the input graph.

For {\sc $(n-k)$-Coloring} it will be useful to consider the complement of our input graph. Let $G$ be a graph and let $\overline{G}$ be its complement. If $M$ is some matching in $\overline{G}$, then for every edge $uv \in M$, it is possible to assign $u$ and $v$ the same color in a proper coloring of $G$. That is, if $|M|\geq k$  then $(G,k)$ is a yes-instance for \textsc{$(n-k)$-Coloring}. Since we are working in the complement of the original graph, we consider the following auxiliary problem \textsc{$(n-k)$-Coloring Complement}.

\begin{framed}
    \noindent
    \textsc{$(n-k)$-Coloring Complement}\\
    \textbf{Input:} A graph $G=(V,E)$ with $n$ vertices and integer $k$.\\
    \textbf{Question:} Does there exist an $(n-k)$-coloring of the complement of $G$?
\end{framed}

We now show that \textsc{$(n-k)$-Coloring Complement} is $\tau$-twin reducible and has positive $m$-bounded matching number, for the functions $m(k)=k-1$ and $\tau(k)=k$. That is, \textsc{$(n-k)$-Coloring Complement} is in TLFPT.

We now give some intuition regarding the $\tau$-twin reducibility of the problem. Note that if there is some vertex which is unique in its color class, then the removal of this vertex decreases both the number of vertices of the graph and the number of colors needed; preserving yes- and no-instances. As any set of open twins in $G$ form a clique in $\overline{G}$ (Observation~\ref{obs: open twins are independent}), each of these vertices must be colored differently. It follows that either at least one of these vertices is unique in its color class, or the graph can be colored with $n-k-1$ colors. That is, it is safe to remove a vertex from this set.

\begin{lemma}\label{lem:color-complement-tau}
    The problem $(n-k)$\textsc{-Coloring Complement} is $\tau$-twin reducible and has positive $m$-bounded matching number, where $\tau(k)=k$ and $m(k)=k-1$.
\end{lemma}
\begin{proof}
    We first show that $(n-k)$\textsc{-Coloring Complement} has $m$-bounded matching number. Let $G$ be some graph, $\overline{G}$ be its complement and let $n = |V|$. Now $(\overline{G},k)$ is an instance of $(n-k)$\textsc{-Coloring Complement}. Recall that $(\overline{G},k)$ is a yes-instance if, and only if, $G$ can be colored with $n-k$ colors.
    
     Suppose there is a matching $M$ of size $k$ in $\overline{G}$. We construct a $(n-k)$-proper coloring of $G$ as follows. Begin by coloring each vertex differently. This is a trivial proper coloring. For each edge $uv \in M$ we now recolor the vertex $v$ giving it the same color as $u$. As $uv$ is an edge in $\overline{G}$ it is a non-edge in $G$, that is, this is a proper coloring of $G$. As for each edge in $M$ this reduces the number of colors used by $1$ and $|M| = k$, this describes an $n-k$ coloring. Thus, $(\overline{G},k)$ is a yes-instance of $(n-k)$\textsc{-Coloring Complement} and $(n-k)$\textsc{-Coloring Complement} has positive $m$-bounded matching number for $m(k) =k-1$.
    
    We continue by showing the problem is $\tau$-twin reducible where $\tau(k)=k$. Suppose $S$ is a set of open twins of size $k + 1$ and let $v$ be one of its elements. We claim that $(\overline{G},k)$ is a yes-instance of $(n-k)$\textsc{-Coloring Complement} if, and only if, $(\overline{G}-v,k)$ is a yes-instance of $(n-k)$\textsc{-Coloring Complement}. As $\overline{G}-v$ has size $n-1$, this is equivalent to showing that there is an $(n-k)$-coloring of $G$ if and only if there is an $(n-k-1)$-coloring of $G-v$.
    
    We first show the backwards direction, that is, if there is an $(n-k)$-coloring of $G$ then there is an $(n-k-1)$-coloring of $G-v$. Let $c \colon V(G)\to [n-k]$ be an $(n-k)$-coloring of $G$. Recall that, by definition, every pair of vertices in $S \cup \{v\}$ are open twins in $\overline{G}$. That is, $S \cup \{v\}$ forms a clique in $G$ and each vertex in $S \cup \{v\}$ is assigned a different color by $c$. We now consider the following three cases.
    \begin{enumerate}
        \item Suppose $v$ is unique in its color class: it follows that $c$ restricted to the graph, $G-v$ uses at most $n-k-1$ colors. 
        \item Assume that $v$ is contained in some color class of size at least $2$. Suppose now that there is some vertex $u \in S$ such that $u$ is unique in its color class. As $u$ and $v$ are twins in $\overline{G}$, they are also twins in $G$. It follows that swapping the colors assigned to $u$ and $v$ we obtain an $n-k-1$ coloring of $G-v$. 
        \item Assume that every vertex in $S \cup \{v\}$ is contained in some color class of size at least $2$. By definition $|S \cup \{v\}| = k+1$. It follows that, if every vertex in $S \cup \{v\}$ is contained in some color class of size at least $2$, then $c$ uses at most $n-k-1$ colors. That is, $c$ is an $n-k-1$ coloring of $G-v$. This concludes the backward direction of our proof.
    \end{enumerate}
    
    We now consider the forward direction. Suppose that there exists an $(n-k-1)$-coloring $c'$ of $G-v$. We extend $c'$ to a coloring of $G$ by assigning $v$ a unique color. By definition, this is a proper coloring of $G$ that uses at most $n-k$ colors. This concludes the proof of our statement. That is, $(n-k)$\textsc{-Coloring Complement} is $\tau$-twin reducible and has positive $m$-bounded matching number, where $m(k)=k-1$ and $\tau(k) = k$.
\end{proof}

By Theorem~\ref{thm:matchTwin}
it follows that $(n-k)$\textsc{-Coloring Complement} is in TLFPT with respect to $k$. To show now that the same holds for \textsc{$(n-k)$-Coloring}, we note the trivial reduction from \textsc{$(n-k)$-Coloring} to $(n-k)$\textsc{-Coloring Complement}: if $G = (V, E)$ is a graph and $\overline{G}$ its complement, by definition $(\overline{G},k)$ is a yes-instance of $(n-k)$\textsc{-Coloring Complement} if, and only if, $(G,k)$ is a yes-instance of \textsc{$(n-k)$-Coloring}. It now remains to show that this is a TLFPT-reduction.

We first show that either we have a trivial instance of $(n-k)$\textsc{-Coloring} or the size of our instance of $(n-k)$\textsc{-Coloring Complement} obtained from our reduction is bounded by the size of our instance of \textsc{$(n-k)$-Coloring}. That is, either we have a trivial instance of $(n-k)$\textsc{-Coloring} or $|\overline{E}| <|E|$. In particular, we show that if $|E| < |\overline{E}|$, i.e. $|E| < \frac{n(n+1)}{2}$, then $(G,k)$ is a yes-instance of \textsc{$(n-k)$-Coloring}. Note that we can count the number of edges in $G$ in time $O(|E|)$ and return a trivial yes instance of $(n-k)$\textsc{-Coloring Complement} if $|E| < \frac{n(n+1)}{4}$. We then assume that, since $\overline{G}$ is sufficiently sparse, $|\overline{E}| <|E|$ and we can build $\overline{G}$ in time $O(|V|+|E|)$. This gives us a TLFPT-reduction from \textsc{$(n-k)$-Coloring} to $(n-k)$\textsc{-Coloring Complement}.

We now argue that if $G$ is sufficiently sparse then $\overline{G}$ contains a matching of size $k$ and so $(G,k)$ is a yes-instance of \textsc{$(n-k)$-Coloring}. To make this argument, we make use of Mader's theorem~\cite{mader1967,mader1968}. Here we state an equivalent theorem; the original statement relates the constant $c_H$ with the average degree of the graph, rather than the number of edges in the graph.
\begin{theorem}[Mader~\cite{mader1967,mader1968}]\label{thm:mader}
    For every graph $H$, there exists a constant $c_H\in\mathbb{N}$ such that any $n$ vertex graph $G$ with at least $c_H\cdot n$ edges contains $H$ as a minor.
\end{theorem}
We use this to show that, if the input graph is sufficiently sparse, it must have a dense complement that contains a matching of at least $k$ edges. Let $c$ be that constant found by Mader~\cite{mader1967}, such that every graph with $n$ vertices and at least $cn$ edges contains the complete graph on $2k$ vertices as a minor and so also a matching of size $k$.

Note that by Lemma~\ref{lem:words-big-as-fn-of-k}, we can assume any function of $k$ is bounded above by $n$. We use this to make the assumption that $k^2\leq\frac{n+1}{4}$.

\begin{lemma}\label{lem-complement-graph-sparse}
    Let $G=(V,E)$ be a graph and let $\overline{G} = (V, \overline{E})$ be the complement of $G$. For any $k$, such that $k^2\leq \frac{n+1}{4}$, either $\overline{G}$ contains a matching of size at least $k$ or $|\overline{E}| \leq |E|$.
\end{lemma}
\begin{proof}
   Mader~\cite{mader1967} shows that $c_{K_r}\leq 8(r-2)\lfloor\log_2(r-2)\rfloor$, where $K_r$ is the complete graph on $r$ vertices. For this problem we note that if this threshold is met for $r=2k$, then the graph contains a matching of size $k$. For ease of writing, let $c =c_{K_r}$ from this point onward.

   Let $G =(V,E)$ be a graph and let $\overline{G}= (V,\overline{E})$ be the complement of $G$. Let $n = |V|$, $m = |E|$ and $\overline{m} = \overline{|E|}$. We now claim that either $\overline{m} \leq m$ or $\overline{m} \geq n \cdot c$. That is, $\overline{G}$ contains a matching of size at least $k$, or it is sparse enough to be computed in linear time. For our purposes, the coarse bound
    \[n \cdot c \leq 8n (2k-2)\lfloor \log_2(2k-2)\rfloor\leq nk^2\]
    is sufficient. We now show that if $\overline{m} >  m$ then $\overline{m} \geq n \cdot c$. Recall our assumption that $k^2\leq \frac{n+1}{4}$. Combining this with the fact that $m+\overline{m}= n(n+1)/2$, if $\overline{m} >  m$ then we find that:
    \begin{align*}
    \overline{m} & > \frac{n(n+1)}{4} \geq n \cdot k^2 \geq n\cdot c.
    \end{align*}

    Thus concluding the proof of this lemma.
\end{proof}

We now consider the case where $|\overline{E}| \geq |E|$. We can now apply the following lemma to obtain the complement of $G$, $\overline{G}$ in time $O(|V|+|E|)$.

\begin{algorithm}
\caption{Complement builder}\label{alg:complement-builder}
\begin{algorithmic}[1]
\Require A graph $G=([n],E')$.
\State Initiate $\overline{G}= ([n],E'=\emptyset)$.
\ForAll{$v\in [n]$}
    \State Initiate $N_{\overline{G}}(v)= [n]$.
    \ForAll{$w\in N_G[v]$}
        \State Remove $w$ from $N_{\overline{G}}(v)$.
    \EndFor
\EndFor
\ForAll{$v\in [n]$}
    \ForAll{$w>v$ in $N_{\overline{G}}(v)$}
        \State Add $(v,w)$ to $E'$.
    \EndFor
\EndFor
\State \Return $\overline{G}$.
\end{algorithmic}
\end{algorithm}

\begin{lemma}\label{lem:complement-time}
    Given a graph $G = (V,E)$, we can build the complement of $G$, $\overline{G} = (\overline{V}, \overline{E})$, in time $O(|V|+|E|+|\overline{E}|)$.
\end{lemma}
\begin{proof}
    In particular we show that Algorithm~\ref{alg:complement-builder} builds $\overline{G}$ in the required time.
    
    The algorithm begins by initiating the graph $\overline{G}$ with the same vertex set as $G$ and no edges. This takes at most $O(n)$ time. Following this, Algorithm~\ref{alg:complement-builder} then generates the neighborhood of each vertex in $\overline{G}$. It does this by, for each vertex $v$, initiating its neighborhood as the entire vertex set, then iterating over all neighbors of $v$ in $G$ and removing them. Initiating the neighborhood requires $O(|V|^2)=O(|E|+|\overline{E}|)$ time. The inner loop requires $O(|E|)$ time since there are two iterations of this loop for each edge. The algorithm then uses the neighborhoods generated to find the edge set of $\overline{G}$. This requires $O(|\overline{E}|)$ time since each edge in $\mathcal{G}$ is added once to $\overline{E}$. It is clear from construction that the graph returned by Algorithm~\ref{alg:complement-builder} is the complement of the input graph. This requires $O(|V|+|E|+|\overline{E}|)$ time in total.
\end{proof}

We now have those components necessary to prove the following theorem.

\begin{theorem}\label{thm: n-k coloring}
     $(n-k)$\textsc{-Coloring} is in TLFPT with respect to $k$.
\end{theorem}
\begin{proof}
    Combining Theorem~\ref{thm:matchTwin} and Lemma~\ref{lem:color-complement-tau}, $(n-k)$\textsc{-Coloring Complement} is in TLFPT. We now prove that $(n-k)$\textsc{-Coloring} is in TLFPT via a TLFPT-reduction from \textsc{$(n-k)$-Coloring} to $(n-k)$\textsc{-Coloring Complement}. Let $(G,k)$, be an instance of $(n-k)$\textsc{-Coloring Complement}, such that $G = (V,E)$. Recall that, by Lemma~\ref{lem:words-big-as-fn-of-k}, we can assume that $k^2\leq\frac{n+1}{4}$.

    We can count the number of edges in $G$ in time $O(|E|)$. If $|E| <n(n+1)/4$, then combining Lemmas~\ref{lem-complement-graph-sparse} and \ref{lem:color-complement-tau} $(G,k)$ is a yes-instance of $(n-k)$\textsc{-Coloring} and our reduction will return some trivial yes-instance of $(n-k)$\textsc{-Coloring Complement}.

    Now suppose that $|E| \geq n(n+1)/4$. This gives us that $|\overline{E}| \leq |E|$. Using Lemma~\ref{lem:complement-time} we construct $\overline{G}$, the complement of $G$, in time $(|V|+|E|)$. By definition, $(G,k)$ is a yes-instance of $(n-k)$\textsc{-Coloring} if and only if $(\overline{G},k)$ is a yes-instance of $(n-k)$\textsc{-Coloring Complement}. That is, there is a TLFPT-reduction from \textsc{$(n-k)$-Coloring} to $(n-k)$\textsc{-Coloring Complement} and so $(n-k)$\textsc{-Coloring} is in TLFPT with respect to $k$.
\end{proof}

\end{toappendix}

\section{TLFPT via Depth and Pruning Kernelization}\label{sec: bounded depth}
In Section~\ref{sec: Matchings and Twins} we placed various problems in TLFPT via the Census Method (Theorem~\ref{thm:matchTwin}). Whereas the Census Method relies on problems having bounded matching number, here we will consider a larger class (c.f.~Corollary~\ref{cor: bounded matching implies bounded depth}) of problems: those with \textit{bounded depth}, as defined below.

\begin{defn}\label{def:l-bounded-depth}
    Let $\Pi$ be some parameterized graph problem. We say that $\Pi$ has \emph{positive (respectively negative) $\ell$-bounded depth} for some function $\ell: \mathbb{N} \to \mathbb{N}$, if, given any instance $(G,k)$ of $\Pi$, whenever $G$ contains a path of length greater than $\ell(k)$, then $(G,k)$ is a positive (resp.~negative) instance of $\Pi$. Note that this trivially holds if no instance of $\Pi$ contains a path of length greater than $\ell(k)$.
\end{defn}

\begin{corollary}\label{cor: bounded matching implies bounded depth}
    If some graph problem $\Pi$ has positive (respectively negative) $m$-bounded matching number, for some function $m: \mathbb{N} \to \mathbb{N}$. Then $\Pi$ has positive (resp.~negative) $\ell$-bounded depth, for the function $\ell(k) = 2m(k)-1$.
\end{corollary}

Any non-trivial instance of a problem with bounded depth must have a short DFS tree. It then follows that such an instance has bounded treedepth \cite{NOdM12}.

\begin{defn} [\cite{NOdM12}]\label{def:treedepth}
The \emph{treedepth} of a graph $G$ is the minimum height of a rooted forest $F$ such that $G \subseteq \text{clos}(F)$ where the closure of $F$ $\text{clos}(F)$ has the vertex set $V(F)$ and the edge set $\{\{x,y\}: x \text{ is an ancestor of }y, x \neq y\}$.
\end{defn}

\begin{lemma}[\cite{NOdM12}]\label{lem:pathtd}
  For a graph~$G$ with a longest path of length $\ell$,
  $\log(\ell) \leq  td(G)\leq \ell$.
\end{lemma}

This allows us to show that {\sc $3$-Coloring} parameterized by treedepth is both twin reducible and has bounded depth.

\begin{theoremrep}\label{thm:3col-td-tlfpt-kern}
    {\sc $3$-Coloring} is in TLFPT parameterized by treedepth.
\end{theoremrep}
\begin{proof}
    Note that from Lemma~\ref{lem:pathtd}, the set of instances containing some path of length greater than $2^k$ is empty. As these instances are trivially negative. {\sc $3$-Coloring} parameterized by treedepth has negative $\ell$-bounded depth, where $\ell(k) = 2^k$.

    We now claim that $\tau$-twin reducible where $\tau(k) = 1$. Consider a graph $G$ and two vertices $u,v \in V(G)$ such that $u$ and $v$ are open twins in $G$. We claim that $G$ is $3$-colorable if, and only if, $G-v$ is $3$-colorable. As any $3$-coloring of $G$ also describes a $3$-coloring of $G-v$, the backwards direction holds. Suppose now we have some $3$-coloring of $G-v$ as $u$ and $v$ are open twins, there is some $3$-coloring of $G$ where $u$ and $v$ are assigned the same color. This proves our claim.
    
    As {\sc $3$-Coloring} parameterized by treedepth has $\ell$-bounded depth and is $\tau$-twin reducible, where $\ell(k) = 2^k$ and $\tau(k) = 1$. It follows from Theorem~\ref{thm:matchTwin} that {\sc $3$-Coloring} parameterized by treedepth admits a TLFPT time kernelization.
\end{proof}

\subsection{Combining Bounded Depth with Twin Subgraphs}\label{sec: bounded depth and twin subgraphs}

We now extend our methods based on bounded depth and twins by generalizing the notion of twin reducibility to \textit{twin subgraphs} (Definition~\ref{def: twin subgraphs}). We use this notion to generalize the equivalence relation defined for open twin vertices to larger isomorphic subgraphs, yielding the notion of \textit{twin subgraphs reducibility} (Definition~\ref{def:tau-twin-subgraph-reducible}).

\begin{defn}\label{def: twin subgraphs}
    We say that a pair subgraphs $H$ and $H'$ of a graph $G$ are twins if:
    \begin{enumerate}
        
        \item there is an isomorphism $i \colon H \to H'$ such that for any $v \in H$, the vertices $v$ and $i(v)$ have the same neighborhood outside of $H \cup H'$; that is $N(v) \setminus V(H \cup H') = N(i(v)) \setminus V(H \cup H')$,
        \item there is no edge $uv \in E(G)$ such that $u \in V(H)$ and $v \in V(H')$.
    \end{enumerate}
    We call $i$ a twin isomorphism.
\end{defn}

\begin{defn}\label{def:tau-twin-subgraph-reducible}
    Let $\Pi$ be a some parameterized graph problem. We say that $\Pi$ is $\tau$\emph{-twin subgraph reducible} for some function $\tau: \mathbb{N}^2 \to \mathbb{N}$ if the sets $YES(\Pi_k)$ and $NO(\Pi_k)$ are closed under the deletion of any subgraph $H$ of size $n$ such that $H$ has at least $\tau(k,n)$ twin subgraphs.
\end{defn}

As for Section~\ref{sec: Matchings and Twins}, the following result yields a general ``Depth and Prune'' method for obtaining TLFPT results by kernelization. Intuitively, we find a bounded height depth-first search tree. Careful use of data-structures allows us to work from the leaves upwards finding and removing subtrees which are twin subgraphs. Due to its bounded height alongside properties of depth first search trees, the resulting graph has bounded size.

\begin{theoremrep}[Depth \& Prune Method]\label{thm:matchTwinSubgraphs}
    Let $\Pi_k$ be some parameterized problem such that for functions $\ell: \mathbb{N} \to \mathbb{N}$ and $\tau: \mathbb{N}^2 \to \mathbb{N}$, $\Pi_k$ is $\tau$-twin subgraph reducible and has $\ell$-bounded depth. $\Pi_k$ admits a linear time kernelization.
\end{theoremrep}
\begin{proof}
	Let $\Pi_k$ be some parameterized graph problem which is both $\tau$-twin subgraph reducible and has $\ell$-bounded depth, for functions $\ell: \mathbb{N} \to \mathbb{N}$ and $\tau: \mathbb{N}^2 \to \mathbb{N}$.
    For any instance $G = (V,E)$ of $\Pi_k$ 
    
    we describe the following $O(|E|) + f(k)$ kernelization algorithm, for some fixed function $f$.
	
	Let $T$ be a depth-first-search tree of $G$ and denote its root by $root(T)$.
    
    We let $depth(T) = max_{v \in V} dist_T(root(T),v)$. If $depth(T) \geq \ell(k)+1$ then $G$ contains a path of length $\ell(k)+1$ and we return the graph $P_{\ell(k)+1}$. As $\Pi_k$ has $\ell$-bounded depth, this is a kernel for $G$. 
    
	We now assume that $depth(T) \leq \ell(k)$. Let $d=depth(T)$.
    Our kernelization algorithm will then proceed by removing subgraphs
    
    of size $r$ having at least $\tau(k,r)$ twin subgraphs. As $\Pi_k$ is size $\tau$- subgraph reducible, the resulting instance is a yes-instance of $\Pi_k$ if and only if $G$ is. To allow this to be done in the necessary time bounds, we will only remove those twin subgraphs admitting an isomorphism that preserves the parent child relations of $T$. Towards this, we define some notation regarding trees. For every $v \in V$, let $parent_T(v)$ denote the parent of $v$. Let $dec_T(v)$ be the set of descendants of $v$ and let $anc_T(v)$ be the set of ancestors of $v$. We write $dec_T[v] := dec_T(v) \cup \{v\}$. 

    We will now define the \textit{type} of a vertex. This will be an equivalence relation on pairs of vertices at the same depth in the tree $T$. Our definition will rest upon an auxiliary notion of \textit{ancestor consistency}, which we define below.
    
    Let $v$ be some vertex, the \textit{ancestor vector} of $v$, $\vec{a}$, is that binary vector of length $depth_T(v)-1$ such that the $i$th element of $\vec{a}$ is $1$ if, and only if, $v$ is adjacent to that vertex with depth $i$ in $anc_T(v)$. We say that a pair of vertices $(v,v')$ are ancestor consistent if they have the same ancestor vector. 
    Note that any two vertices with the same ancestor vector necessarily have the same depth.

    We are now define types recursively: two vertices $v$ and $v'$, both at some given depth $r$, are of the \underline{same type} if: they are ancestor consistent and have the same number of children of each type. 
    
    We illustrate the concept of types with a small example in \cref{fig:dfs-types}.

    \begin{figure}
        \centering
        \includegraphics[width=0.8\linewidth, page=8]{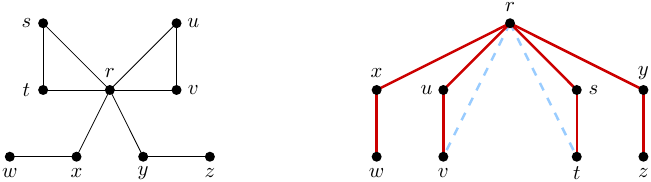}
        \caption{An 9-vertex graph $H$ (left), and a DFS tree $T$ of $H$ (right). Edges of $H$ which are not edges of $T$ are shown in dashed light blue. In the DFS tree shown:
        \begin{itemize}
            \item $v$ and $t$ are of the same type, because they have the same ancestor vector, i.e. $\vec{a}(v)=\vec{a}(t)=(1)$ (both vertices are adjacent to $r$), and neither has any children. Call their type $\alpha$.
            \item Analogously, $w$ and $z$ are of the same type, because $\vec{a}(w)=\vec{a}(z)=(0)$. Call their type $\beta$.
            \item Each of $x, u, s, y$ has the same (empty) ancestor vector: however, seeing as $x$ and $y$ each have a single child of type $\beta$ whereas $u$ and $s$ each have a single child of type $\alpha$. Consequently, $x$ and $y$ are of the same type, which is different from the shared type of $u$ and $s$.
            \item The root $r$ has its own type. As we shall see, the type of $r$ ``encodes'' the entire structure of the graph.
        \end{itemize}}
        \label{fig:dfs-types}
    \end{figure}

    \begin{claim}
        Let $u$ and $v$ be vertices with the same type. If $parent_T(u) = parent_T(v)$, then $G[dec_T[u]]$ and $G[dec_T[v]]$ are twin subgraphs.
    \end{claim}
    \begin{claimproof}
        As $u$ and $v$ have the same type, they have the same depth in $T$. Let $r = depth_T(u) = depth_T(v)$. By definition, every pair of vertices with the same type have the same number of children of each type. It follows that, as $u$ and $v$ have the same type, there is a  bijection $i: dec_T[u] \to dec_T[v]$ that preserves both vertex types and parent–child relations. That is, for every vertex $c \in dec_T[u]$, $c$ and $i(c)$ have the same type
        and $parent_T(c) = i^{-1}(parent_T(i(c))$.

        We now claim that $i$ is a twin isomorphism. That is, for every edge $(c,c') \in E(G)$ such that $c \in dec_T[u]$, either: \begin{enumerate*}
            \item $c' \in dec_T[u]$ and $(i(c),i(c')) \in E(G)$, or
            \item $c' \in V \setminus (dec_T[u] \cup dec_T[v])$ and $(i(c),c') \in E(G)$
        \end{enumerate*}. The case where $c \in dec_T[u]$ will follow symmetrically.
        
        As $T$ is a depth first search tree, $c'$ is that unique vertex in $anc_T(c)$ with depth $depth_T(c')$. As $anc_T(c) \cap dec_T[c] = \emptyset$, $c' \notin dec_T[v]$. We first consider case (i), that is, $c' \in dec_T[u]$. As $parent_T(w) = parent_T(i(w))$, for every $w \in dec_T[u]$, $i(c')$ is the unique vertex in $anc_T(i(c))$ with depth $depth_T(c')$. By definition of $i$, $c$ and $i(c)$ have the same type and so are also ancestor consistent. It follows that the edge $i(c)i(c') \in E(G)$. 
        
        We now consider consider case (ii), that is, $c' \in V \setminus (dec_T[u] \cup dec_T[v])$. Given that $parent_T(u) = parent_T(v)$, it follows that $c'$ is that unique vertex in $anc_T(i(c))$ with depth $depth_T(c')$. As $c$ and $i(c)$ are ancestor consistent, $i(c)i(c') \in E(G)$.
    \end{claimproof}

    We can partition the vertices of $G$ into \textit{equivalence classes} such that any pair of vertices belong to the same equivalence class if, and only if, they have the same type. For every equivalence class $C$, by definition there exist constants $r$ and $s$ such that every vertex $v$ in $C$ has depth $r$ and $|dec_T[v]| =s$. We say that the class $C$, and its corresponding type, has depth $r$, subtree size $s$ and capacity $\tau(k,s)$.

    We say $G$ is $\tau$-pruned with respect to $T$ if, for every vertex $v$ and equivalence class $C$, $children_T(v) \cap C \leq \tau(k,s)$ where $s$ is the subtree size of $C$.

    \begin{claim}
        There exist functions $eq: \mathbb{N}^3 \to \mathbb{N}$ and $maxSize: \mathbb{N}^2 \to \mathbb{N}$, such that if $G$ is $\tau$-pruned with respect to $T$, then: (1) for every $r \in \{1,\ldots, d\}$ there are at most $eq(k,d,r)$ equivalence classes with depth $r$ and (2) $G$ has size at most $maxSize(k,d)$.
    \end{claim}
    \begin{claimproof}
        We define the function $eq$ recursively. In the process of this, we also determine the subtree size associated with each type, which will later allow us to define $maxSize$.

        Consider first a vertex $v$ of depth $d$. By definition, $v$ is a leaf, and so its type is determined only by its ancestor vector. Since there are $2^{d-2}$ possible ancestor vectors for a vertex of depth $d$, it follows that $eq(k,d,d) = 2^{d-2}$. Moreover, every type at depth $d$ has subtree size $1$.

        We now define $eq(k,d,r)$ for $r \in \{1,\ldots,d-1\}$. We first fix some total ordering of the $eq(k,d,r+1)$ types of vertices at depth $r+1$, and let $m_i$ and $s_i$ denote the capacity and subtree size, respectively, of the $i$th type of vertex with depth $r+1$. Since $G$ is $\tau$-pruned, any vertex of depth $r$ has at most $m_i$ children of type $i$. That is, for each type $i$, there are $m_i+1$ possible choices for the number of children of type $i$, and hence the total number of distinct such choices is $\prod_{i=1}^{eq(k,d,r+1)} (m_i+1)$.

        Further,  there are $2^{r-1}-1$ possible ancestor vectors for a vertex of depth $r$. Combining these observations, $eq(k,d,r) = (2^{r-1} -1)\cdot \prod_{i=1}^{eq(k,d,r+1)} (m_i + 1)$.
        
        Let $v$ be a vertex of depth $r$ with exactly $c_i$ children of type $i$ for each $i \in \{1,\ldots,eq(k,d,r+1)\}$. Recall that the $i$th type of vertex with depth $r+1$ has subtree size $s_i$, that is, for every vertex $w$ with this type $dec_T[w] = s_i$. That is, the subtree size of the equivalence class containing $v$ is $1 + \sum_{i=1}^{eq(k,d,r+1)} c_i \cdot s_i$. 
        
        Finally, we define $maxSize(k,d)$. We again fix a total ordering of the $eq(k,d,2)$ types of vertices at depth $2$, and let $m_i$ and $s_i$ denote the capacity and subtree size, respectively, of the $i$th type. Then $G$ has size at most $maxSize(k,d) = \sum_{i=1}^{eq(k,d,2)} m_i \cdot s_i$.
    \end{claimproof}

     We define the \textit{$\tau$-pruning} of $G$ as follows. Beginning with the leaves, we process vertices in order of decreasing depth. For a vertex $v$, if there exist at least $\tau(k, |dec_T[v]|)$ vertices of the same type as $v$ and share the same parent in $T$, then remove the entire subtree $dec_T[v]$ from the graph.

    As $\Pi_k$ is $\tau$-twin subgraph reducible, the $\tau$-pruning of $G$ is a kernel for $G$. It remains to show that the $\tau$-pruning of $G$ can be computed in time $O(|E(G)|) + f(k)$, for some function $f$ depending only on $k$.

     Our algorithm will, for each vertex in $G$, return a ID corresponding to its type in the $\tau$-pruning of $G$ (supposing it has not been removed). That is, there is a bijection between the set of IDs assigned to vertices of $G$ and the equivalence classes of the $\tau$-pruning of $G$. The ID of $v$ consists of two parts. The first part, which we call the \textit{ancestor ID}, encodes the ancestor vector of $v$ and the second, which we call the \textit{subtree ID}, encodes the number of children $v$ has of each type. To build intuition regarding the subtree ID, we begin by defining the \textit{subtree vector} of $v$.

     For each $r \in \{1,\ldots,d\}$, we fix a total order on the $eq(k,d,r$) equivalence classes with depth $r$, and refer to the $i$th class under this order as the $i$th equivalence class with depth $r$. 
     
     The subtree vector of $v$, $\vec{s}$, is a vector of length $eq(k,d,depth_T(v)-1)$ with each element of $\vec{s}$ corresponding to the number of children that $v$ has in each equivalence class. For $i \in \{1, \ldots, eq(k,d,depth_T(v)-1)\}$, let $C_i$ be the $i$th equivalence class with depth $depth_T(v)-1$ and let $m_i$ be the capacity of $C_i$. Let $$\vec{s}[i] := min(|children_T(v)\cap C_i|, m_i).$$

     For each $v \in V$, let $\vec{a}_v$ denote the ancestor vector of $v$ and let $\vec{s}_v$ denote the subtree vector of $v$. Since $\sum_{v\in V} |\vec{a}_v|+|\vec{s}_v| = f(k) \cdot |V|$, for some function $f$, it is not possible to naively allocate the memory for these vectors in TLFPT time. Instead, as both the length of these vectors and the value of their elements is bounded by some function of $k$, Lemma~\ref{lem:words-big-as-fn-of-k} will allow us to assume that $k$ is small enough that each vector can be encoded within a single machine word of size $\log n$. As we assign two words for each vertex in $G$ this requires only time $O(|V|)$.

     We now explicitly define the encoding of these vectors and give a bound on $k$ which will be sufficient for our purposes.
     The ancestor ID, $ancID(v) := \sum_{i \in \{1, \ldots, |\vec{a}_v|\}} \vec{a}_v[i] \cdot 2^i$. As $\vec{a}_v$ is a binary vector whose final element is $1$, the ancestor vector is uniquely determined by $ancID(v)$ and $ancID(v) \leq 2^{depth_T(v)-1}+1$. 
     Likewise, let $m_{max}$ denote the maximum capacity of any equivalence class of depth $depth_T(v)+1$, we define the subtree ID of $v$ by $subtreeID(v) := \sum_{i \in \{1, \ldots, |\vec{s}|\}} \vec{s}[i] \cdot (m_{max}+1)^i$. This construction results in a bijection between the possible subtree vectors and the subtree IDs that a vertex with depth $r$ could take. Further, $subtreeID(v) \leq (m_{max}+1)^{eq(k,d,depth_T(v)+1)}$. 
     We then let the ID of $v$ be $ID(v) = ancVec(v) + subtreeID(v) \cdot 2^{|\vec{a}|}$. If $m_{root}$, is the maximum capacity of any equivalence class of depth $2$, then for every vertex $v$, $ID(v) \leq (m_{root}+1)^{eq(k,d,2)}$. That is, applying Lemma~\ref{lem:words-big-as-fn-of-k} we assume that $(m_{root}+1)^{eq(k,d,2)} \leq \log n$ and so operations on these IDs can be done in constant time.

    Recall that the ID of $root(T)$ encodes the $\tau$-pruning of $G$. We can now begin to compute the ID of each vertex of $G$. To this end, we introduce the following auxiliary functions, which we precompute and store in arrays so as to allow constant-time access thereafter.

    Let $subtreeIDtoVec: \mathbb{N}^4 \to \mathbb{N}^{\mathbb{N}}$ be the function such that a vertex with depth $r$ and subtree ID $i$ has subtree vector $subtreeIDtoVec(k,d,r,i)$. Let $IDtoType: \mathbb{N}^3 \to \mathbb{N}$ be the function such that a vertex with ID $i$ belongs to the $IDtoType(k,d,i)$th equivalence class (with the same depth). Let $typetoCapacity: \mathbb{N}^4 \to \mathbb{N}$ be that function such that the $i$th equivalence class with depth $r$ has capacity $typetoCapacity(k,d,r,i)$.
    
    Since these functions will always be taken with the same first two arguments, $k$, $d$, to simplify notation we suppress these parameters in the following. We compute the value $eq(k,d,r)$, for every $r \in \{1, \ldots, d\}$ and store this in an array $eq$ such that $eq[r] = eq(k,d,r)$. Each of these values can thereafter be accessed in constant time. Likewise, we compute the vector $subtreeIDtoVec(k,d,r,i)$, for each possible subtree ID $i$ which could appear at depth $r$, and store this in a (two dimensional) array $subtreeIDtoVec$ such that $subtreeIDtoVec[r,i] = subtreeIDtoVec(k,d,r,i)$. We compute the value $IDtoType(k,d,i)$, for each possible ID $i$ and store this in an array $IDtoType$ such that $IDtoType[i] = IDtoType(k,d,i)$. We compute the value $typetoCapacity(k,d,r,i)$, for every $i \in \{1, \ldots, eq(k,d,r)\}$ and store this in a (two dimensional) array $typetoCapacityt$ such that $typetoCapacity[r,i] = typetoCapacity(k,d,r,i)$. Finally, we compute the array $maxCapacity$ such that $$maxCapacity[r] = max_{i \in\{1, \ldots, eq(k,d,r)\}} typetoCapacity[r,i].$$  
    By definition $d \leq \ell(k)$, that is these arrays can be constructed in time $f(k)$, for some fixed function $f$.

    \begin{algorithm}
        \caption{id(v)}\label{alg:extDFS}
        \begin{algorithmic}[1]
        \Require Graph $G$; DFS tree $T$ of $G$ with depth $d$; some vertex $v$ with depth $r$ in $T$; arrays $eq$, $subtreeIDtoVec$, $IDtoType$, $typetoCapacity$, $maxCapacity$.
        \State Let $\vec{f}_a$ be a vector of length $r-1$ such that $\vec{f}_a[i]=2^i$ for each $0 \le i < r-1$.
        \State Let $\vec{f}_s$ be a vector of length $eq(r+1)$ such that $\vec{f}_s[i]=(maxCapacity[r+]+1)^i$ for each $0 \le i < eq[r+1]$.        
        \Comment{We construct $\vec{f}_a$ and $\vec{f}_s$ only once for each depth $r \in \{1, \ldots, d\}$.}

        \State Let $ancID = 0$.
        \State Let $subtreeID = 0$.
        \ForAll{$u \in N(v)$}
            \If {$depth_T(u) = r +1$} \Comment{$u$ is a child of $v$}
                \State $subtreeVec = subtreeIDtoVec[r, subtreeID]$
                \State $type_u = IDtoType[r+1,id(u)]$
                \If{$subtreeVec[type_u] \leq typetoCapacity[r+1,type_u]-1$}
                    \State $subtreeID = subtreeID + \vec{f}_s[type_u]$ \Comment{increment count of children of $type_u$}
                \EndIf
            \EndIf
            \If {$depth_T(u) < r-1$} \Comment{$u$ is an ancestor of $v$} 
            	\State $ancID = ancID + \vec{f}_a[depth_T(u)]$.
            \EndIf
        \EndFor
        \State \Return $ancID + subtreeID(2^{r-1})$
        \end{algorithmic}
    \end{algorithm}

    We note that for every $v \in V$, Algorithm~\ref{alg:extDFS} returns $ID(v)$. Further, Algorithm~\ref{alg:extDFS} is called on each vertex of $V$ at most once (assuming memoization). A single iteration of $id(v)$ requires $O(|N(v)|)$ time, that is $id(root(T))$ can be computed in $O(|E|)$ time. It follows that in total our kernelization algorithm takes $O(|E|)+f(k)$, for some function $f$.
\end{proof}

We now turn to three classical problems: {\sc $k$-Path},
{\sc $k$-Vertex Ranking},
and {\sc Satisfiability}.

\begin{framed}
    \noindent
    \textsc{$k$-Path}\\
    \textbf{Input:} graph $G=(V,E)$,\\
    \textbf{Parameter:} $k$,\\
    \textbf{Question:} does $G$ contain $P_k$ as a subgraph?
\end{framed}

\begin{theorem}\label{thm:kpath-tlfpt-kern}
    {\sc $k$-Path} admits a TLFPT time kernelization.
\end{theorem}
\begin{proof}
	We first claim that {\sc $k$-Path} parameterized by $k$ is $\lfloor \frac{k}{2} \rfloor$-twin subgraph reducible and has $k-1$-bounded depth.

	Let $(G, k)$ be an instance of {\sc $k$-Path}. By definition, if $G$ contains a path of length $k$ then $(G,k)$ is a yes-instance of {\sc $k$-Path}, that is {\sc $k$-Path} has $k-1$-bounded depth.

    We now claim that {\sc $k$-Path} is $\lfloor \frac{k}{2} \rfloor$-twin subgraph reducible. To show this, let $S$ be a set of pairwise twin subgraphs, let $P$ be some path of length $k$ in $G$ and let $Z \subseteq S$ be those subgraphs containing some vertex from $P$. By definition, any path between distinct subgraphs $H, H' \in Z$ must contain some vertex in $G \setminus S$. It follows that $|Z| \leq \lfloor \frac{k}{2} \rfloor$ meaning {\sc $k$-Path} is $\lfloor \frac{k}{2} \rfloor$-twin subgraph reducible.
\end{proof}

A \textit{$t$-ranking} of a graph $G$, for some integer $t$, is a $t$-coloring of $G$ such that, for every pair of vertices $x$ and $y$ with $c(x) = c(y)$ and for every path between $x$ and $y$, there is a vertex $z$ on this path with $c(z) > c(x)$. The rank of $G$, is the smallest $t$ for which $G$ admits a $t$-ranking.
The \textsc{$k$-Vertex Ranking} problem is defined as follows.

\begin{framed}
    \noindent
    \textsc{$k$-Vertex Ranking}\\
    \textbf{Input:} graph $G=(V,E)$,\\
    \textbf{Parameter:} $k$,\\
    \textbf{Question:} does there exist a $k$-ranking of $G$?
\end{framed}

The \textsc{$k$-Vertex Ranking} problem is also interesting as it is known that the treedepth of a graph is equal to its vertex ranking number and also to the minimum height of any elimination tree of $G$ (see \cite{NOdM12}, Chapter 6 for a fuller discussion of treedepth). It follows that, by showing $k$-\textsc{Vertex Ranking} is in TLFPT, the same also holds for computing treedepth.

We also show that \cref{thm:matchTwinSubgraphs} can be leveraged for problems over objects which are not necessarily graphs. We showcase this with the classical problem of boolean {\sc Satisfiability}, which asks whether a given Boolean formula $\psi$ in conjunctive normal form (CNF) is satisfiable. In our choice of parameter, we aim to capture the structural properties of an instance, particularly the relationships between its variables and clauses. We do this by encoding the formula as a graph. The \textit{incidence graph}, $I(\psi)$, of a formula $\psi$, is that bipartite graph such that one part contains the variables of $\psi$ and the other part contains the clauses of $\psi$. Two vertices are adjacent to one another in $I(\psi)$ if one of these vertices corresponds to a clause and the other corresponds to some variable that appears in this clause. We will focus in particular on the treedepth of the incidence graph, which we call its incidence treedepth.
\begin{toappendix}

\begin{thm}\label{thm:treedepth-tlfpt-kern}
    {\sc $k$-Vertex Ranking} admits a TLFPT time kernelization.
\end{thm}
\begin{proof}
    We first claim that {\sc $k$-Vertex Ranking} parameterized by $k$ has $2^{k+1}$-bounded depth and is $\tau$-twin subgraph reducible where $\tau(k,i)= (k-1) \cdot k^{i} +1$ for $i \geq 1$. The class of graphs which admit a $k$-ranking are closed under taking minors (Lemma 6.2 in \cite{NOdM12}) and a path of length $l$ (on $l+1$ vertices) has rank $\lceil \log_2(l+2) \rceil$ (Eqn. 6.2 of \cite{NOdM12}).

    Let $(G, k)$ be an instance of {\sc $k$-Vertex Ranking}. If $G$ contains a path of length $2^{k+1}$, then $(G,k)$ has rank at least $k+1$, hence, {\sc $k$-Vertex Ranking} has $2^k$-bounded depth. We now claim that {\sc $k$-Vertex Ranking} is $\tau$-twin subgraph reducible where $\tau(k,i)= (k-1) \cdot k^{i} +1$ for $i \geq 1$. Let $S$ be some subgraph and let $T$ be a set such that $T \cup \{S\}$ is a set of pairwise twin subgraphs. By definition of twin subgraphs, each of these subgraphs has the same size, say they each have size $n$. We now claim that if $|T| = (k-1) \cdot k^{n}+1$, then there is a $k$-ranking of $G - V(S)$ if, and only if there is a $k$-ranking of $G$.

    If there is some $k$-ranking of $G$, by definition, this is also a $k$-ranking of $G - V(S)$. Let us now consider the backwards direction. Suppose that there is some $k$-ranking $c$ for $G - V(S)$. We now describe a coloring $c'$ of $G$ and claim that this is a $k$-ranking. Towards this, we say that a pair of subgraphs $X,Y \in T \cup \{S\}$ have the same coloring under $c$, if that neighborhood preserving isomorphism between $X$ and $Y$ also preserves the colors of $c$. As there are $k^{n}$ possible colorings of the subgraphs of $T \cup \{S\}$ and $|T| = (k-1) \cdot k^{|S|}+1$, it follows that there is some set $E \subseteq T$ such that $|E| \geq k$ and the subgraphs of $E$ each have the same coloring under $c$. We obtain $c'$ from $c$ by coloring the vertices of $S$ such that, $S$ has the same coloring as every subgraph in $E$ under $c'$.

    We now claim that $c'$ is a $k$-ranking of $G$. Assume for the sake of contraction that $c'$ is not a $k$-ranking of $G$, that is, there is some path $P = (p_1, \ldots, p_\ell)$ such that $c'(p_1)=c'(p_\ell)$ and, for every $i \in \{1, \ldots, \ell\}$, $c'(p_i) \leq c'(p_1)$. As $c$ is a $k$-ranking of $G -V(S)$, $P$ is not a path in $G -V(S)$. It follows that $V(S) \cap V(P) \neq \emptyset$.
    
    We claim that $P$ contains vertices from at most $k$ different subgraphs in $T \cup \{S\}$. This will allow us to reason that there is some subgraph $X$ in $E$ which does not contain any vertices in $P$. This will allow us to map the path $P$ to a path $P' = (p'_1, \ldots, p'_\ell)$ such that $V(S) \cap V(P') = \emptyset$ and $c'(p'_i) = c'(p_i)$, for every $i \in \{1, \ldots, \ell\}$. As $P'$ is a path in $G- V(S)$ we can conclude that $c$ was not a $k$-ranking of $G- V(S)$. 

    We will now show that $P$ contains vertices from at most $k$ different subgraphs in $T \cup \{S\}$. By definition of twin subgraphs, for every $S' \in T$, $N(S) \setminus V(S) = N(S') \setminus V(S')$. Let $N$ denote the set $N(S) \setminus V(S)$.
    It follows that $G - V(S)$ contains the complete bipartite graph $K_{|N|, |T|-1}$ as a subgraph. Given that $K_{|N|, |T|-1}$ has rank $\min (|N|, |T|-1)+1$ (this is a folklore result with an easy proof). 
    As $|T| > k$ and $G - V(S)$ has rank at most $k$, it follows that $|N| \leq k-1$. Combining this with the fact that, by definition of twin subgraphs, the subgraphs of $T$ are connected components in $G - N$, we find that $P$ contains vertices from at most $k$ different subgraphs in $T$. That is, there is some subgraph $X \in E$ which does not contain any vertices from $P$.

    We are now ready to construct that path $P'$. Let $i: V(S) \to V(X)$ be that neighborhood preserving isomorphism between $S$ and $X$, by definition this isomorphism also preserves the color of vertices. Let $P' = (p'_1, \ldots, p'_\ell)$ be the path such that for $j \in \{1, \ldots, \ell\}$, $p'_j = i(p_j)$, if $p_j \in V(S)$, and $p'_j = p_j$ otherwise. As $P'$ contains no vertices from $S$ and $c'(p_j)=c'(p'_j)$ for each $j \in \{1, \ldots, \ell\}$, it follows that $c$ was not not a $k$-ranking $G \setminus S$. As this is a contradiction, we find that $c'$ is a $k$-ranking of $G$. 
    
    It follows that there is a $k$-ranking of $G \setminus S$ if, and only if there is a $k$-ranking of $G$ and so {\sc $k$-Vertex Ranking} is $\tau$-twin subgraph reducible where $\tau(k,i)= (k-1) \cdot k^{i} +1$ for $i \geq 1$.
\end{proof}
\end{toappendix}

\begin{toappendix}
    
We now give the definition of \textsc{Satisfiability}.

\begin{framed}
    \noindent
    \textsc{Satisfiability}\\
    \textbf{Input:} Boolean formula $\psi$ in conjunctive normal form (CNF).\\
    \textbf{Question:} is $\psi$ satisfiable?
\end{framed}

 Many parameterizations of SAT aim to capture the structural properties of an instance, particularly the relationships between its variables and clauses. A common way to express this structure is to encode the formula as a graph and analyze structural parameters of that graph. Notable examples include the treewidth of the associated graph, as well as the parameter known as \textit{maximum deficiency}, as introduced in \cite{franco2003perspective}, which is based on matchings in this graph.

There are several ways to encode a SAT instance as a graph.

\apxblue{Recall that the \textit{incidence graph}, $I(\psi)$,  of $\psi$ is that bipartite graph such that one part contains the variables of $\psi$ and the other part contains the clauses of $\psi$. Two vertices are adjacent to one another in $I(\psi)$ if one of these vertices corresponds to a clause and the other corresponds to some variable that appears in this clause. We will focus on the incidence graph and in particular the treedepth of the incidence graph, which we call its incidence treedepth.}

\end{toappendix}

\begin{toappendix}
    
It will be useful to consider the following variant of the incidence graph which we will call the \textit{literal incidence graph}. For a CNF formula $\psi$, the literal incidence graph, $I_L(\psi)$, of $\psi$ is that bipartite graph such that one part contains the clauses of $\psi$ and one part contains the literals of $\psi$. That is, for each variable $x$ in $\psi$, $I_L(\psi)$ contains a pair of vertices, one corresponding to $x$ and the other to $\neg x$. The following lemma will allow us to work on the literal incidence graph.

\begin{lemma}\label{lem:SAT-var-to-lit}
	For every CNF formula $\psi$, $td(I_L(\psi)) \leq 2 \cdot td(I(\psi))$. 
\end{lemma}
\begin{proof}
	Let $\psi$ be some CNF formula. To simplify notation, let $G_I$ be the incidence graph of $\psi$ and let $G_L$ be the literal incidence graph of $\psi$. Let $T= (V(G_I), E_T)$ be some minimal treedepth decomposition of $G_I$. That is, every edge in $G_I$ is between an ancestor and a descendant in $T$. Further, $T$ has depth $td(G_I)$. We now give a treedepth decomposition of $G_I$, $T'$, with height at most $2 \cdot td(G_I)$. For every variable $x$ in $\psi$, let $parent_{T'}(x) = parent_{T}(x)$, $parent_{T'}(\neg x) = x$ and $children_{T'}(\neg x) = children_{T}(x)$. Note that $T'$ has depth at most $2\cdot depth(T)$. It remains to show that $T'$ is indeed a treedepth decomposition for $G_L$. This follows from the observation that, for every variable in $\psi$, $anc_{T'}(\neg x) \setminus \{x\} =  anc_{T}(x)$ and $dec_{T'}(\neg x) =  dec_{T}(x)$. As every edge in $G_I$ is between an ancestor and a descendant in $T$, every edge in $G_L$ is between an ancestor and a descendant in $T'$. It follows that $td(G_L) \leq 2 \cdot td(G_I)$.
\end{proof}

\begin{theorem}\label{thm:SAT}
    {\sc Satisfiability}, when parameterized by its incidence treedepth, admits a TLFPT time kernelization.
\end{theorem}
\begin{proof}
    To apply those techniques developed previously, we will define a graph problem taking as input the literal incidence graph of our formula. We note that there is a unique formula corresponding to any literal incidence graph\footnote{To allow literals and clauses in $V(G_L)$ to be distinguished, we shall assume that there is an odd number of clauses (possibly by introducing a dummy clause containing some literal and its negation). Note that there is always an even number of literals.}, and vice versa. Due to this, it will be useful to consider the following auxiliary problem called {\sc Literal Incidence Graph Satisifability} or {\sc LIGSat} for short. 
    \begin{framed}
    \noindent
    \textsc{LIGSat}\\
    \textbf{Input:} A literal incidence graph $G_L$. \\ 
    \textbf{Question:} is the formula corresponding to $G_L$ satisfiable?
    \end{framed}
    We will now show that {\sc LIGSat} is in TLFPT when parameterized by treedepth. This will then allow us to show that \textsc{Satisfiability} parameterized by incidence treedepth is in TLFPT via a TLFPT reduction to {\sc LIGSat} parameterized by treedepth.

    We first show that {\sc LIGSat} is in TLFPT when parameterized by treedepth. In particular, we show that this problem is $\tau$-twin subgraph reducible and has $\ell$-bounded depth, for functions $\tau(k,n) = 2^n$ and $\ell(k) = 2^k$.

    From Lemma~\ref{lem:pathtd}, if a graph treedepth at most $k$ then it contains no path of length greater than $2^k$. That is, the set of instances of our problem containing a path of length at least $2^k$ is empty. That is, trivially, {\sc LIGSat} parameterized by treedepth has negative $\ell$-bounded depth, where $\ell(k) = 2^{2k}$.

    \begin{claim}
        \textsc{LIGSat} parameterized by treedepth is $\tau$-twin reducible.
    \end{claim}
    \begin{proof}
    Let $(G_L,k)$ be some instance of {\sc LIGSat} parameterized by treedepth.
    Let $H$ be some subgraph of $G_L$ and let $S$ be a set of subgraphs such that $S \cup \{H\}$ consists of pairwise twin subgraphs. By definition of twin subgraphs (\cref{def: twin subgraphs}), every graph in $S \cup \{H\}$ has the same size, which we shall denote by $n$. 
    
    Our task is to prove that if $|S| \geq 2^n$, then $(G_L,k)$ is a yes-instance of {\sc LIGSat}, if, and only if, $(G_L-V(H),k)$ is. We let $\psi$ denote the formula corresponding to $G_L$ and let $\psi'$ be that formula corresponding to $G_L-V(H)$.

    First suppose that $(G_L,k)$ is a yes-instance, and denote by $a$  some satisfying assignment for $\psi$. To simplify arguments, we let $a$ be an assignment over the literals of $\psi$. As $S \cup \{H\}$ consists of pairwise twin subgraphs, for any pair of subgraphs $X$, $Y$ in $S \cup \{H\}$ there is a twin isomorphism $i$ from $X$ to $Y$. We note that as $G_L$ is bipartite, $i$ must map clauses to clauses and literals to literals. We say that $X$ and $Y$ have the same assignment under $a$ if there exists a twin isomorphism $i$ from $X$ to $Y$ such that, for every literal $l$ in $X$, $a(l) = a(i(l))$. 
    Recall that every subgraph in $S \cup \{H\}$ has size $n$. Thus, the number of literal vertices in any such subgraph is at most $n$, and the number of possible truth assignments to these literal vertices is at most $2^n$.
    As $|S \cup \{H\}| \geq 2^n+1$, there must exist some pair of subgraphs in $S \cup \{H\}$ with the same assignment under $a$. Without loss of generality these are $H$ and some other subgraph $X \in S$. 
    
    Let $i_H$ be that twin isomorphism from $H$ to $X$. We now claim that $a$ is a satisfying assignment for $\psi'$. 
    Consider some clause $C$ in $\phi$ which is satisfied by some literal $l$ in $H$. Then either $C$ is in $H$, in which case $C$ does not appear in $\phi'$, or $i_H(l) \in C$, since $H$ and $X$ are twin subgraphs. In the latter case, $C$ is satisfied by the literal $i_H(l)$, since $X$ has the same assignment as $H$ under $a$. Thus, $a$ is a satisfying assignment for $\phi'$.

    We now consider the other direction. Suppose that there is some satisfying assignment $a'$ for $\psi'$. We extend this to an assignment $a$ for $\psi$ as follows. For some arbitrary subgraph $X \in S$, we take a twin isomorphism $i$ from $X$ to $H$ and we assign values to the literals of $H$ such that $X$ and $H$ have the same assignment under $a'$ (i.e. for each literal $l$ in $H$ $a(l)=a(i(l))$). 
    
    We now claim that $a$ is a satisfying assignment for $\psi$. 
    By assumption, every clause in $\psi'$ is satisfied by $a'$ and thus by $a$, which extends it. Let us now consider some clause $C$ which is in $\psi$ but not in $\psi'$, that is $C$ is in $H$. As $a'$ is a satisfying assignment $\psi'$, the clause $i(C)\in X$ is satisfied by some literal $l'$ in $G-V(H)$. 
    If $l' \notin X$ then $l'\in C$ by the definition of a twin subgraph isomorphism and $C$ is satisfied. Otherwise, $l' \in X$ and the literal $l = i^{-1}(l')$ is in $C$ and is also true under $a$, and again $C$ is satisfied. Thus $a$ satisfies every clause of $\psi$ and the claim follows.
    \end{proof}
    
    That is, we have now shown both directions and $(G_L,k)$ is a yes-instance, if, and only if, $(G_L-V(H),k)$ is a yes-instance. It follows that {\sc LIGSat} parameterized by treedepth is $\tau$-twin reducible and has $\ell$-bounded depth, for functions $\tau(k,n) = 2^n$ and $\ell(k) = 2^k$. Applying Theorem~\ref{thm:matchTwinSubgraphs}, {\sc LIGSat} parameterized by treedepth is in TLFT.

    We now show that there is a TLFPT reduction from \textsc{Satisfiability} parameterized by incidence treedepth to {\sc LIGSat} parameterized by treedepth. The literal incidence graph can be constructed from a formula in time linear in the number of variables and clauses, and conversely, the formula can be reconstructed from the graph within the same time bound. Further, by Lemma~\ref{lem:SAT-var-to-lit}, the treedepth of a literal incidence graph is at most $2$ times larger than the incidence treedepth of its corresponding formula. Thus, there is a TLFPT reduction from \textsc{Satisfiability} parameterized by incidence treedepth to {\sc LIGSat} parameterized by treedepth. It follows that \textsc{Satisfiability} parameterized by incidence treedepth is in TLFPT, thus concluding the proof of our theorem.
\end{proof}
\end{toappendix}

\section{TLFPT via Dynamic Programming: Parameterizing by BFS-Width}\label{sec:BFS-w}

In \cref{sec:degree-based-kern,sec: Matchings and Twins,sec: bounded depth}, we gave TLFPT results via kernelization. We now turn to another workhorse of parameterized complexity: dynamic programming.
It may seem that dynamic programming algorithms align more naturally with the multiplicative (rather than additive) definition of FPT. Maintaining a state space bounded by a function of $k$, a classical algorithm would need to expend $f(k)$ operations for each of the $O(n)$ parts the decomposition processed, yielding a runtime on the order of $f(k) \cdot O(n)$.
In fact, we show that by carefully considering our model of computation we can also apply dynamic programming with a TLFPT runtime.

In this section, we parameterize graph problems by \textit{BFS-width}, a parameter which in some sense mirrors treedepth: graphs of bounded treedepth have ``nice'' DFS trees, a structure we leverage in \cref{sec: bounded depth}, whereas graphs of bounded BFS-width have ``nice'' BFS trees. 
The techniques we apply here could plausibly be lifted to other decompositions (e.g. tree decompositions), but the question of \textit{producing} an optimal tree decomposition 

in time TLFPT parameterized by treewidth remains open. 
BFS-width, on the other hand, has the desirable property that producing a decomposition in linear time is straightforward.
This natural parameter was only recently formally studied from a parameterized perspective \cite{EGL25BFSW}.

\begin{definition}[BFS-decomposition, BFS-width, \cite{EGL25BFSW}]
    For a graph $G=(V,E)$, a BFS-decomposition of $G$ rooted at vertex $r$ is a partition of $V$ into layers $L_0, \ldots, L_d$ satisfying $L_i=\{u \in V : \mathrm{dist}(r,u)=i\}$. 
    The \emph{width} of a BFS-decomposition is the size of its largest layer. 
    The BFS-width of a graph is the maximum width of any BFS-decomposition it admits. 
\end{definition}

In this section, we show that many

classic problems can be solved in TLFPT time when parameterized by BFS-width, using algorithms that we call BFS-nice, defined below. 
We note that any BFS-decomposition is also a path decomposition (and hence a tree decomposition) and that many algorithms operating on a path or tree decomposition are also BFS-nice.  

At a high level, for a given (slice of a parameterized) problem $\Pi_k$, our approach consists of:
\begin{enumerate}
    \item Describing a ``BFS-nice'' dynamic programming algorithm $\mathcal A$ solving the problem $\Pi_k$ in time $O(n) \cdot f(k)$. (Note: $\mathcal{A}$ is never run on the input.) \label{step:bfs-nice-alg}
    \item Describing an algorithm $\mathcal DC$ which decomposes and compresses the input graph $G$ in time $O(n) + f(k)$. \label{step:decomp-and-compress}
    \item By simulating $\mathcal A$ (in time bounded by some function of $k$), ``compile'' it into a Deterministic Finite Automaton-like representation: obtain a state space (including a starting state), accepting states, and transition function which together describe the behavior of $\mathcal{A}$ on an input graph decomposition. \label{step:transition-from-simulation}
    \item Run $\mathcal DC$ on the input graph $G$ to obtain its compressed representation. Then, apply the compiled transition function of $\mathcal A$ on the compressed representation of $G$ in $O(n)$ time and return YES if $\mathcal{A}$ would. \label{step:apply-transition-func}
\end{enumerate}

\begin{defn}[Bicolored induced subgraph]
    Let $G=(V,E)$ be a graph and $A,B \subseteq V$ be disjoint sets of vertices. The \emph{bicolored induced subgraph} of $G$ with respect to $A$ and $B$ is the induced subgraph $G[A \cup B]$ together with a (not necessarily proper) 2-coloring of its vertices assigning color $0$ to all vertices in $A$ and color $1$ to all vertices in $B$. When the bicoloring is clear from context, we use the shorthand ``bicolored graph $G[A \cup B]$'' to mean the bicolored induced subgraph of $G$ with respect to $A$ and $B$.  
\end{defn}

Before formally presenting the main result of the section, we need to detail the steps above.

\FloatBarrier

\subsection{Step \ref{step:bfs-nice-alg}: What is a BFS-nice algorithm?}\label{subsec:bfs-nice-alg}

We take this opportunity to note that many algorithms operating on a path decomposition are also BFS-nice, and that a BFS-decomposition is also a path decomposition. 

\begin{defn}\label{defn:nice-algorithm-bfs-w}
    A dynamic programming algorithm $\mathcal A$ which takes as input a graph $G$ and a width-$k$ BFS-decomposition $L_0, \ldots L_d$ of $G$, and which decides parameterized problem $\Pi$ is \emph{BFS-nice} if there exists a function $f$ such that:
    \begin{itemize}
        \item $\mathcal A$ maintains an internal state $\sigma$ (initially $\sigma_0$, which is independent of the input but possibly dependent on $k$) which never exceeds $f(k)$ bits in size,
        \item For each $i \in \{1, \ldots, d\}$, $\sigma_i$ is computed as a function of $\sigma_{i-1}$ and the bicolored graph $G[L_i \cup L_{i-1}]$ alone, and
        \item $\mathcal{A}$ returns YES or NO as a function of $\sigma_d$ alone.
    \end{itemize}
    A parameterized problem $\Pi$ is BFS-nice if there exists a BFS-nice algorithm $\mathcal{A}$ solving $\Pi$.
\end{defn}

This definition captures many natural dynamic programming algorithms, including the one given in \cite{FG06PCT} 
to solve $K_3$-\textsc{Coloring} with a given a tree decomposition. We generalize this slightly and describe a BFS-nice algorithm for the more general $H$-\textsc{Coloring} (defined just below), proving that it is BFS-nice.
\begin{framed}
    \noindent
    \textsc{$H$-Coloring}\\
    \textbf{Input:} graph $G$.\\
    \textbf{Question:} does there exist a homomorphism from $G$ to $H$, that is, a map $m:V(G)\to V(H)$ such that $(u,v) \in E(G) \implies (m(u),m(v)) \in E(H)$?
\end{framed}

\begin{lemma}\label{lem:h-coloring-bfs-nice}
    For any fixed graph $H$, $H$-\textsc{Coloring} admits a BFS-nice algorithm.
\end{lemma}
\begin{proof}
\cref{alg:basic-h-coloring} is a simple dynamic program which decides $H$\textsc{-Coloring} for an input graph and its $k$-width BFS-decomposition. This is nothing revolutionary; cf.~\cite{FG06PCT} Example 11.35, which solves $K_3$-\textsc{Coloring} given a tree decomposition of width at most $k$.
It is easy to verify that this algorithm is BFS-nice: the internal state which here is denoted by $\mathcal L_i$ is of size bounded by a function of $k$ (namely, the number of assignments of $[k] \to V(H)$ -- recall that $H$ is part of the problem definition) and is updated exclusively as a function of the previous internal state and the bicolored graph $G[L_i \cup L_{i-1}]$, and the algorithm returns YES or NO as a function of its final state $\mathcal L_d$. 
\end{proof}

\begin{algorithm}[!ht]
\caption{BFS-nice $H$-\textsc{Coloring} algorithm}\label{alg:basic-h-coloring}
\begin{algorithmic}[1]
\Require Graph $G=(V,E)$ on $n$; integers $k$ and $d$; BFS-decomposition of $G$ into layers $L_0 \ldots L_d$ each of size at most $k$; graph $H$ with which to color $G$. 
\State Compute all $|V(H)|$ labelings of the root layer $L_0$ with $V(H)$. Denote $\mathcal L_0$ the set of these (partial) solutions (i.e. labelings) any of which is locally admissible for the vertices in bag $L_0$ (since $L_0$ consists of a single vertex, that vertex can be mapped to any vertex in $H$).
\ForAll{$i \in \{1, \ldots d\}$}
\State Compute all possible labelings of layer $L_i$. 
Denote these by $\mathcal L_i$.
\ForAll{$\ell_i \in \mathcal L_i$}
\If{$\nexists \ell_{i-1}\in \mathcal L_{i-1}$ such that $\ell_{i-1} \cup \ell_{i}$ is a an $H$-coloring of $G[L_i \cup L_{i-1}]$}
\State Remove $\ell_i$ from $\mathcal L_i$
\EndIf
\EndFor
\EndFor 
\If{$\mathcal L_d = \emptyset$}
\Return NO
\Else ~\Return YES
\EndIf
\end{algorithmic}
\end{algorithm}

\FloatBarrier

\subsection{Step \ref{step:decomp-and-compress}: Compressing small graphs into registers}\label{subsec:decomp-and-compress}

Leveraging the fact that we may assume that the word size $w$ (the number of bits which can be stored in a register) exceeds any $f(k)$ for any function $f$ we choose (\cref{lem:words-big-as-fn-of-k}), we can fit encodings of mathematical objects of size bounded by some function of $k$ in a single register. 
For example, a simple labeled $k$-vertex graph might be represented by its adjacency matrix using $k^2$ single-bit registers. The analogous ``compressed'' representation is a single $k^2$-bit word. Actually, a $\binom{k}{2}$-bit word suffices, since the diagonal is all-zero and the matrix is symmetric. As shown in the figure, the $\binom{k}{2}$-bit representation of a $k$-vertex graph $H$ on $k$ vertices can be constructed as the sum of $|E(H)|$ terms in $[2^{\binom{k}{2}}]$. We define \texttt{kPunch} to be the lookup table storing these terms, indexed by the adjacencies they encode. 
Analogously, a $2k$-vertex balanced bipartite graph $H$ can be represented as the sum of $|E(H)|$ terms from $[2^{k^2}]$, and we define \texttt{kkPunch} the lookup table storing the set of integers $[2^{k^2}]$ indexed by the adjacencies they encode. 

We provide a little more illustration of this idea now. 
A graph $G=(V,E)$ on $k$ vertices and $m$ edges with $E=\{e_1,\ldots,e_m\}$ can be obtained as the union of $m$ graphs, namely $(V,E)=\cup_{1 \le i \le m}(V,\{e_i\})$. Fixing $k$, we can denote by $P_{i,j}=([k], \{(i,j)\})$ the \textit{punch graph} encoding the edge $(i,j)$ for $k$-vertex graphs. Then any $k$-vertex graph $G$ with $m$ edges is the union of $m$ different punch graphs. The adjacency matrix of $G$ is correspondingly the sum of the adjacency matrices of $m$ different punch graphs. The table \texttt{kPunch} exactly stores the (compressed) adjacency matrices of all punch graphs on $k$ vertices. This is illustrated in \cref{fig:graph-to-word}.

\begin{figure}[!ht]
    \centering
    \includegraphics[width=\linewidth, page=7]{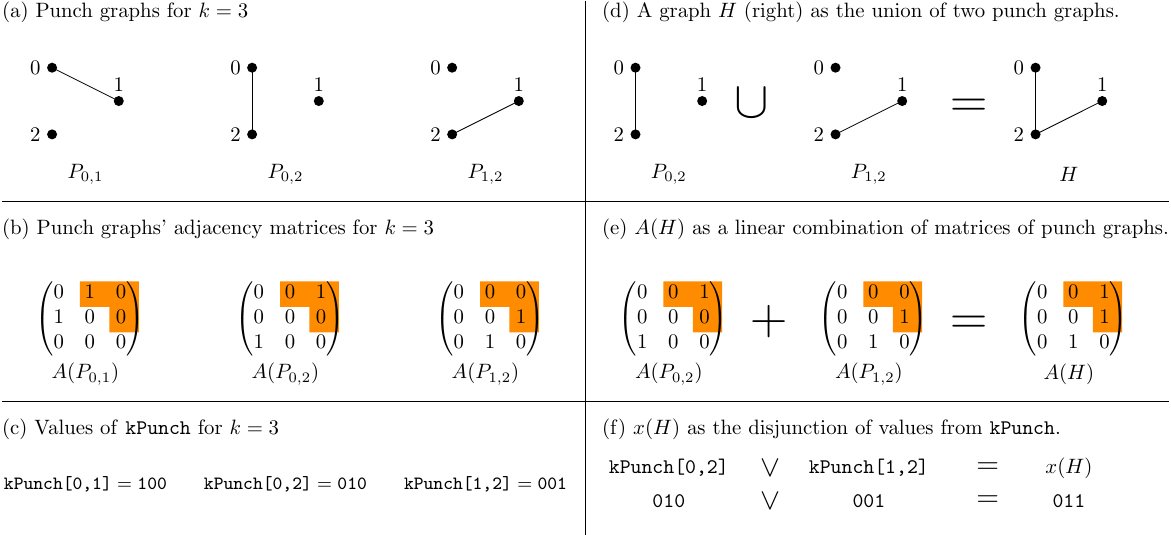}
    \caption{Illustration of punch graphs and the lookup table \texttt{kPunch} when $k=3$. The word $x(H)=\texttt{011}$ represents the 3-vertex graph $H$ with $V(H)=\{0,1,2\}$ and $E(H)=\{(0,2), (1,0)\}$. More generally, for any graph $H$, $b(H)$ can be computed as
    $x(H)= \bigvee_{(u,v) \in E(H)} \texttt{kPunch}[u,v]$.}
    \label{fig:graph-to-word}
\end{figure}

Likewise, a balanced bigraph $G=([k] \cup [k+1,2k], E)$ on $2k$ vertices and $m$ edges can be obtained as the union of $m$ bigraphs on $2k$ vertices and $1$ edge. The graph $G$ and each of these one-edge bigraphs can be represented by a $k^2$-bit word (the compressed version of its biadjacency matrix). Again, the $k^2$ bit word representing $G$ can be computed as the sum of the $m$ different $k^2$-bit ``punch'' words. Punch words for balanced $2k$ vertex bigraphs are stored in the lookup table \texttt{kkPunch}, indexed by the adjacencies which they encode, so that \texttt{kkPunch}$[i,j]$ is the $k^2$-bit word encoding the bigraph with an edge from $i$ to $k+j$.

Our algorithm $\mathcal DC$ (\cref{alg:bfsw-decomp-and-compress}) takes as input a graph $G$ together with a vertex $r$ and returns a compressed representation of the entire graph. This is done by first creating its BFS decomposition rooted at $r$ and then computing a special symbol $b_i$ for each consecutive pair of layers of the decomposition so that $b_i$ fits in a single register but fully describes both layers and their relationship to one another. In particular, $b_i$ encodes: the size of each layer; the adjacency matrix of each layer; and the biadjacency matrix between the two layers.
This is illustrated in \cref{fig:bfs-w-compression-i,fig:bfs-w-compression-ii}. In that example the tuple $(|L_i|, |L_{i-1}|, x_i, x_{i-1}, y_i)=(2,3,\texttt{011},\texttt{100},\texttt{101010000})$ fully describes the two layers shown. Noting that $|L_i|=3=\texttt{11}_2$ and $|L_{i-1}|=2=\texttt{10}_2$, it follows that the word $b_i = \texttt{1110011100101010000}$
exactly encodes the two layers as well.

\begin{figure}[!ht]
    \centering
    \includegraphics[width=\linewidth, page=10]{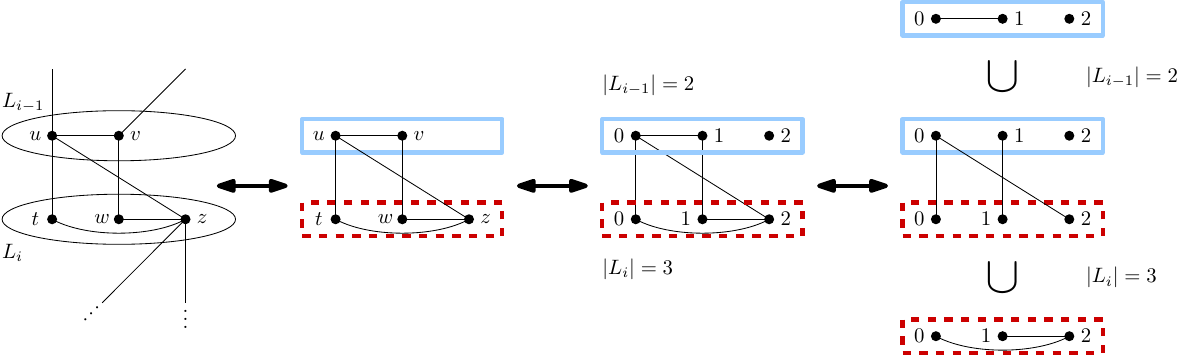}
    \caption{Illustrating the ideas used in \cref{alg:bfsw-decomp-and-compress} for $k=3$. Left: two consecutive layers of a decomposition. Center left: a bicolored graph representing the two layers. Center right: the inclusion of the layers' sizes allows the addition of a dummy vertex, so each graph is on $k$ vertices. Right: the graph is the union of a bigraph on $2k$ vertices and two graphs on $k$ vertices each.}
    \label{fig:bfs-w-compression-i}
\end{figure}

\begin{figure}[!ht]
    \centering
    \includegraphics[width=.6\linewidth, page=11]{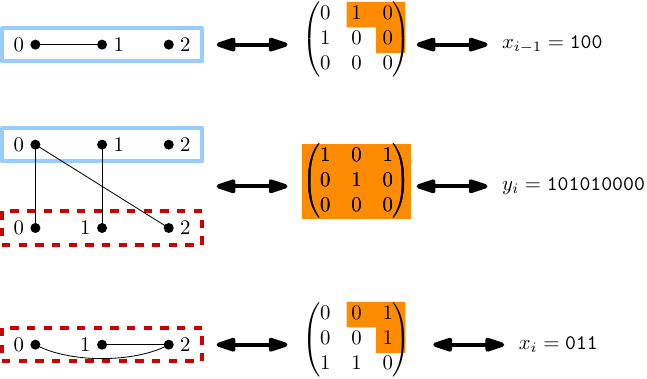}
    \caption{Illustrating how the two graphs and the bicolored bigraph from \cref{fig:bfs-w-compression-i} can be encoded are encoded binary strings.}
    \label{fig:bfs-w-compression-ii}
\end{figure}

In our algorithm, too, $x_i$ is used to denote that word encoding the internal adjacency matrix of each layer, and $y_i$ that word encoding the biadjacency matrix between any two layers.

\begin{algorithm}[!ht]
\caption{Decomposition and Compression algorithm ($\mathcal DC$)}\label{alg:bfsw-decomp-and-compress}
\begin{algorithmic}[1]
\Require Graph $G=(V,E)$ on $n$ vertices; vertex $r$ which roots a BFS tree of width at most $k$ in $G$. 
\State Compute the BFS tree rooted at $r$. Denote by $d$ the depth of the tree, and by $L_0,\ldots,L_d$ its layers such that $L_0=\{r\}$. \label{dc-line:bfs-tree} 
Define the following functions using lookup tables: \begin{itemize}
    \item $l:V \to [0,d]$ which returns the depth of a vertex (the index of the layer it belongs to)
    \item $\mathrm{index}:V \to [0,k-1]$ which returns the index of a vertex in its layer
    \item $|L_i|$ which returns the number of vertices in a layer $i$. 
\end{itemize} 
\State Initialize the two-dimensional arrays \texttt{kPunch} and \texttt{kkPunch} as defined above. \label{dc-line:init-punches}
\ForAll{$i \in [d]$} \label{dc-line:init-templates}
    \State Set $x_i = 0$ and $y_i=0$.
\EndFor
\ForAll{$(u,v)\in E$} \label{dc-line:loop-over-edges}
\State $i,~j = \mathrm{index}(u), \mathrm{index}(v)$ 
\If{$l(u)==l(v)$} \Comment{$u$ and $v$ are in the same layer.}
\State $x_{l(u)} = x_{l(u)} \lor~$\texttt{kPUNCH[i,j]} \Comment{Bitwise logical OR}
\Else \Comment{$|l(u)-l(v)|=1$ -- adjacent layers.}
\If{$l(u)<l(v)$} swap $u$ and $v$ and swap $i$ and $j$. \EndIf \Comment{Ensure $l(u)>l(v)$} 
\State $y_{l(u)} = y_{l(u)} \lor~$\texttt{kkPUNCH[i,j]}
\EndIf
\EndFor
\ForAll{$i \in \{1, \ldots, d\}$}\label{dc-line:loop-over-layers}
\State Do $b_i = |L_i|$. \Comment{Encoding size of layer $L_{i}$}
\State Do $b_i = b_i << {\lfloor \log_2 k \rfloor + 1}$ \Comment{Left-shift by $\lfloor \log_2 k \rfloor + 1$ bits.} 
\State Do $b_i = b_i \lor |L_{i-1}|$. \Comment{Encoding size of layer $L_{i-1}$}
\State Do $b_i = b_i << {\binom{k}{2}}$ \Comment{Left-shift by $\binom{k}{2}$ bits.}
\State Do $b_i = b_i \lor x_i$ \Comment{Encoding adjacencies within layer $i$.}
\State Do $b_i = b_i << {\binom{k}{2}}$ \Comment{Left-shift by $\binom{k}{2}$ bits.}  
\State Do $b_i = b_i \lor x_{i-1}$ \Comment{Encoding adjacencies within layer $i-1$.}
\State Do $b_i = b_i << {k^2}$ \Comment{Left-shift by $k^2$ bits.}  
\State Do $b_i = b_i \lor y_i$ \Comment{Encoding adjacencies between layers $i$ and $i-1$.}
\EndFor
\State \Return $b_1, \ldots, b_d$. \Comment{Each word $b_i$ fully encodes the bicolored graph $G[L_i \cup L_{i-1}]$.}
\end{algorithmic}
\end{algorithm}

\begin{lemma}\label{lem:bfsw-dc-runtime-and-correctness}
    \cref{alg:bfsw-decomp-and-compress} runs in time $O(n) + f(k)$ and fully encodes the graph $G$ into $d$ integers $b_1, \ldots, b_d$. 
\end{lemma}

\begin{proof}
    Line \ref{dc-line:bfs-tree} can be done in linear time. A breadth-first search takes $O(|V|+|E|)$ time \cite{CLRS}; the tables defined each are of size linear in that of the input, and are indexed by the vertex set $V$ (in the case of $l$ and $\mathrm{index}$ and $d$ (in the case of $|L_i|$ both of which are bounded by $n$.\\
    Line \ref{dc-line:init-punches} takes time $f(k)$.\\
    The loop on line \ref{dc-line:init-templates} take time $d$ in total.\\
    Line \ref{dc-line:loop-over-edges} iterates exactly $|E|$ times. Each iteration of the loop takes constant time.\\
    Line \ref{dc-line:loop-over-layers} iterates $d$ times. Each iteration consists of a single bitwise operation which can be done in constant time. 
    It is clear from the description of the algorithm each word $b_i$ returned encodes the bicolored graph $G[L_i \cup L_{i-1}]$, as illustrated in \cref{fig:bfs-w-compression-i,fig:bfs-w-compression-ii}.
\end{proof}

\FloatBarrier

\subsection{Step \ref{step:transition-from-simulation}: state space and transition function of \texorpdfstring{$\mathcal{A}$}{A}}\label{subsec:transition-from-simulation}

We had alluded earlier to the internal state $\sigma$ of $\mathcal{A}$ and its properties. Recall in particular: 

\begin{itemize}
    \item $\sigma$ is initially $\sigma_0$, which is independent of the input but possibly dependent on $k$
    \item $\sigma$ never exceeds $f(k)$ bits in size for some (possibly non-computable) function $f$
    \item At each layer $i$ of the BFS decomposition, $\sigma_i$ is computed as a function of $\sigma_{i-1}$ and the bicolored graph $G[L_i \cup L_{i-1}]$  alone.
    \item $\mathcal{A}$ returns YES or NO as a function of the final state $\sigma_d$ alone.
\end{itemize}

We denote by $\mathcal{B}_k$ the set of bicolored (not necessarily bipartite) graphs with at most $k$ vertices of each color. These correspond to the possible values that a bicolored graph $G[L_i \cup L_{i-1}]$ can take. Also, we denote by $\Sigma_k$ the set of possible values of $\sigma_i$ over all inputs to $\mathcal{A}$ and integers $i$. Note that $\max_{\sigma \in \Sigma_k} |\sigma| \le f(k)$. We shall compute the transition function $\varsigma : \Sigma_k \times \mathcal{B}_k$ which exactly describes how $\sigma_i$ depends on $\sigma_{i-1}$ and the bicolored graph $G[L_i \cup L_{i-1}]$ in the execution of $\mathcal{A}$. 

This is very straightforward if $f$ is computable; simply enumerate all structures on at most $f(k)$ bits (since $\Sigma_k$ is necessarily a subset of this) and simulate $\mathcal{A}$ on tuples consisting of such a structure and a bicolored graph from $\mathcal{B}_k$. However, if $f$ is not computable, we can instead compute $\varsigma$ using (fittingly) breadth-first search, by simply maintaining a fringe of all the states reachable from $\sigma_0$ by $\mathcal{A}$; since $f(k)$ is finite, it follows that $\Sigma_k$ is finite, and our BFS will eventually terminate. 
Thus we can compute, with respect to any fixed $k$ which bounds the BFS-width of the input of $\mathcal{A}$: the state space $\Sigma_k$ of $\mathcal{A}$; the set of accepting states (in the finite-state automaton sense) $\Sigma_k^{\mathrm{YES}}$; and its transition function $\varsigma$.

The transition function $\varsigma: \Sigma_k \times \mathcal{B}_k \to \Sigma_k$ is then ``compressed'' (in the same sense as we compress $G$ above): we make use of the bijection described in \cref{alg:bfsw-decomp-and-compress} which encodes bicolored graphs as single words (or equivalently integers), which we denote $b:\mathcal{B}_k \to \mathbb \mathbb{N}$, and also fix an arbitrary bijection $s:\Sigma_k \to [|\Sigma_k|]$. Note that both the domain of $b$ $\mathrm{dom}(b)$ and the domain of $s$ $\mathrm{dom}(s)$ are (interpretable as) sets of bitstrings of length bounded by a function of $k$, namely ${k^2 + 2\left (\lfloor \log k \rfloor + 1 + \binom{k}{2}\right)}$ in the case of $b$ and $\lfloor \log |\Sigma_k| \rfloor + 1$ in the case of $s$. 

Now any tuple $(b_i,s_i) \in \mathrm{dom} (b) \times \mathrm{dom}(s)$ can be encoded as a single bitstring $b_i \bconcat s_i$ obtained by left-shifting $b_i$ by $\lfloor \log |\Sigma_k| \rfloor + 1$ bits and summing the result with $s_i$. Thus, the rightmost $\lfloor \log |\Sigma_k| \rfloor + 1$ bits of $b_i \bconcat s_i$ are equal to $s_i$ and the right-shift of $b_i \bconcat s_i$ by $\lfloor \log |\Sigma_k| \rfloor + 1$ bits is exactly $b_i$. Note that:
\begin{itemize}
    \item $b_i \bconcat s_i$ is of size bounded by a function of $k$ (and so at most the word size by \cref{lem:words-big-as-fn-of-k}).
    \item each of $b_i$ and $s_i$ can be retrieved in constant time from $b_i \bconcat s_i$
    \item $b_i \bconcat s_i$ can be computed in constant time from $b_i$ and $s_i$
\end{itemize}
This enables us to encode the transition function $\varsigma$ as a lookup table \texttt{transition} (again, of size bounded by a function of $k$) indexed by $\mathrm{dom}(b) \times \mathrm{dom}(s)$. Then \texttt{transition}$[b(B) \bconcat s(\sigma)] = s(\varsigma(B,\sigma))$ for each $B\in \mathcal{B}_k$, $\sigma \in \Sigma_k$.
Also, we can compress the starting state $\sigma_0$ as \texttt{start} = $s(\sigma_0)$.
Lastly, we produce a table \texttt{accepting} holding all elements of $\{s(\sigma) : \sigma \in \Sigma_k^\mathrm{YES}\}$. Note that checking for membership of \texttt{accepting} can be done in time $f(k)$.

\FloatBarrier

\subsection{Step \ref{step:apply-transition-func}: applying \texorpdfstring{\texttt{transition}}{transition} over \texorpdfstring{$(b_1,b_2,\ldots,b_d)$}{b1,b2,...,bd}}\label{subsec:apply-transition-func}

We can now simulate $\mathcal{A}$ on input graph $G$ by applying our \texttt{transition} table and the compressed encoding $(b_1,b_2,\ldots,b_d)$ of $G$ which is output by \cref{alg:bfsw-decomp-and-compress}. The algorithm is very simple:

\begin{algorithm}[!ht]
\caption{Applying the \texttt{transition} over the compressed input.}\label{alg:bfsw-appling-transition}
\begin{algorithmic}[1]
    \Require Graph $G=(V,E)$, vertex $r$ of $G$ which roots a BFS tree with layers $L_0, \ldots L_d$ each of size at most $k$, BFS-nice algorithm $\mathcal{A}$.
    \State Obtain $(b_1,b_2,\ldots,b_d)$ by running \cref{alg:bfsw-decomp-and-compress} from \cref{subsec:decomp-and-compress}. \label{lst:line:run-dc}
    \State Obtain the word \texttt{start}, tables \texttt{transition} and \texttt{accepting}, and functions $s$, $b$ and compressing its outputs as detailed in \cref{subsec:transition-from-simulation} above.\label{lst:line:simulating-a}
    \State Do \texttt{state} = \texttt{start}
    \ForAll{$1 \le i \le d$} \label{lst:line:state-update-loop}
        \State Do \texttt{state} = \texttt{transition}[$b_i \bconcat $\texttt{ state}] \label{lst:line:state-update-statement} 
    \EndFor
    \If{\texttt{state} $\in$ \texttt{accepting}} 
    \Return YES \label{lst:line:check-accepting}
    \Else \State \Return NO \EndIf
\end{algorithmic}
\end{algorithm}

\begin{lemma}\label{lem:bfsw-applying-transition}
    \Cref{alg:bfsw-appling-transition} terminates in time $O(n)+f(k)$ and returns YES on an input graph $G$ with BFS decomposition $L_0,\ldots, L_d$ if and only if $\mathcal{A}$ returns YES on the same input.
\end{lemma}
\begin{proof}
    The runtime of line \ref{lst:line:run-dc} is $O(n+f(k))$ as already mentioned in \cref{lem:bfsw-dc-runtime-and-correctness}. \\
    As described in the \cref{subsec:transition-from-simulation} above, Line \ref{lst:line:simulating-a} terminates in time bounded by some function of $k$, since it does not depend on the input $G$, only on $k$ and on the algorithm $\mathcal{A}$ which is fixed. \\
    The loop at line \ref{lst:line:state-update-loop} iterates $d \le n$ times, and the execution of line \ref{lst:line:state-update-statement} takes constant time: each of computing $b_i \bconcat$ \texttt{state} a lookup in \texttt{transition} take constant time.\\
    To check whether \texttt{state} appears in the table \texttt{accepting} on line \ref{lst:line:check-accepting} takes time bounded by a function of $k$, since both objects are bounded in size by such a function.\\
    The correctness of the algorithm is straightforward to verify; on the $i$th iteration of the loop on line \ref{lst:line:state-update-loop} the value of \texttt{state} is equal to $s(\sigma_i)$ where $\sigma_i$ is the actual internal state which algorithm $\mathcal{A}$ would have immediately after processing the $i$th layer of the BFS-decomposition. The algorithm then returns YES if and only if $\mathcal{A}$ would return YES with $G$ and $L_0,,\ldots,L_d$ as input. 
\end{proof}

We can now state our main result.

\begin{theoremrep}\label{thm:bfs-nice-gives-tlfpt}
    If $\mathcal{A}$ is a BFS-nice algorithm solving parameterized problem $\Pi$, then there exists an algorithm $\mathcal{A'}$ solving $\Pi$ in time $O(n)+f(k)$. 
\end{theoremrep}
\begin{proof}
    Follows directly from \cref{alg:bfsw-appling-transition} and \cref{lem:bfsw-applying-transition}.
\end{proof}

\begin{corollary}[Of \cref{thm:bfs-nice-gives-tlfpt} and \cref{lem:h-coloring-bfs-nice}]
    For any fixed graph $H$, \textsc{$H$-Coloring} is in TLFPT parameterized by the BFS-width of the input graph.
\end{corollary}

We take this opportunity to note that, interestingly, \textsc{3-Coloring} is then TLFPT parameterized by either treedepth (\cref{thm:3col-td-tlfpt-kern}) or BFS-width.

\FloatBarrier

\section{Discussion and Future Work}
Inspired by the realities of modern computational needs, ours is an early exploration of truly linear FPT. Here distinctions between the additive and multiplicative definitions of FPT make a difference: LFPT -- $O(n) \cdot f(k)$ -- strictly contains TLFPT -- $O(n) + f(k)$. 

Towards being of service to practical computation, we have concentrated on the \textit{positive toolkit} of TLFPT. Via appropriate parameterizations and paying attention to the data structures involved in linear-time algorithmics, we placed many problems in TLFPT including: {\sc SAT, Vertex Cover, Min-Max Matching, $(n-k)$-Coloring, Diverse Pair of Matchings, $k$-Path, $k$-Vertex Ranking} and {\sc $H$-Coloring}.

In the case of {\sc $H$-Coloring}, where we parameterized by BFS-width and used dynamic programming, the similarity to parameterizations by tree-width is evident. So we ask: 
\begin{question}
    Is MSO$_2$ model checking in TLFPT parameterized by BFS-width?
\end{question}

Turning to {\sc $k$-Path}, our result is a first step towards being able to integrate TLFPT and the graph minors machinery; thus we ask: 
\begin{question}\label{q:minors}
    For which families $\{H_i\}_{i \in {\mathbb N}}$ is $H_i$-minor TLFPT parameterized by $|H_i|$?
\end{question}

Although not our focus here, the negative toolkit is equally important. By diagonalization, we have shown that there exist parameterized problems in LFPT (and thus in FPT) that are not in TLFPT. However, still lacking natural TLFPT-hard problems and techniques for proving TLFPT-hardness, we identify the following direction for further work.
\begin{question}
    Which FPT problems are likely not to be in TLFPT?
\end{question}
 
As we mentioned in Section~\ref{sec:BFS-w}, it seems plausible to generalize our dynamic programming results to classes of ``bounded width''. 
Of course, this would be useful only if we have access to the corresponding decompositions. We ask:

\begin{question}\label{q:pathwidth-tlfpt}
    Is pathwidth TLFPT parameterized by pathwidth? 
\end{question}

\begin{question}\label{q:treewidth-tlfpt}
    Is treewidth TLFPT parameterized by treewidth? 
\end{question}

\begin{question}\label{q:rankwidth-tlfpt}
    Is rankwidth TLFPT parameterized by rankwidth? 
\end{question}

There is no shortage of directions for further research. Consider, for example,  programmatic questions such as determining how TLFPT fares in practice, or how to design TLFPT approximation algorithms. In general, the theory of TLFPT is incipient and we expect its development to spark healthy conversation between theory and practice.

\newpage

\bibliography{ref}
\appendix

\section{Extended discussion of word-RAM architecture}\label{sec:why-word-ram-appendix}

Because we are interested in truly linear computation, the details of accessing and updating information about a graph are very important.  We also therefore need to be clear on our model of computation.  We use the word-RAM model \cite{FW90-STOC,FW90-FOCS} because it is the canonical model used for the analysis and design of fine-grained algorithms \cite{HHRTT24, EvdHM24}, though this is sometimes left implicit.  

A RAM machine \cite{CR73} classically consists of a finite program operating on an infinite sequence of registers, each with the capacity to hold an arbitrary integer. There are two core versions of the RAM model: one in which reading (or writing) the integer $n$ to a register, or accessing the $n$th register, takes constant time (``unit-cost RAM''), and another in which the same operations take time $\log |n|$. The former is liable to abuse of the constant-time assumption (e.g. by generating operands of size polynomial in the original input size, each of which can be stored in a single register), whereas the latter is limited compared to the actual hardware algorithms are typically executed on, which can indeed access registers (up to some bound) in constant time.

The word-RAM model \cite{FW90-STOC,FW90-FOCS} is a so-called trans-dichotomous model -- it aims to bridge the gap between these two paradigms. 
    The operation set (and the cost of any operation) is the same as for the unit-cost RAM model, but with the additional constraint that each register can only hold words consisting of up to $b$ bits, with $b=\log_2 \lceil N \rceil$ where $N$ is the size of the program input. Thus, arithmetic operations with operands of size at most that of the input can be performed in constant time, but (consistent with the reality that computers have fixed word length) the model forbids ``computations that achieve hidden parallelism by doing operations on `long words'.'' \cite{FW90-FOCS} 

Importantly, we will see that it is possible for a word to be long with respect to a parameter $k$ while still being sublinear with respect to the overall input size $n$. This reflects a sound assumption: that if $k$ is sufficiently small with respect to $n$, the (real-world) architecture which enables constant-time operations with operands of size $n$ can be leveraged to operate on ``long words'' of size $f(k) \le n$.

Recall that $w = \log n$ and that we may assume $n >> k$ (if $n$ is upper-bounded by any function of $k$ then we can simply apply an arbitrary algorithm which decides the problem at hand; the runtime of this algorithm on such an input is then also bounded by a function of $k$). 
\end{document}